\documentclass[journal,a4paper]{IEEEtran}

\IEEEoverridecommandlockouts

\usepackage{graphicx,cite}
\usepackage{times}
\usepackage{amssymb}
\usepackage{amsmath}
\usepackage{amsfonts}
\usepackage{url}

\usepackage{caption}
\usepackage{subfig}

\usepackage{dsfont}
\usepackage{bbm}

\usepackage{algorithm}
\usepackage{algcompatible}

\usepackage{verbatim}

\usepackage{tikz}

\newtheorem{corollary}{\bf Corollary}
\newtheorem{theorem}{\bf Theorem}

\newtheorem{definition}{\bf Definition}

\newcommand*{\myfigfactor}{0.5}
\newcommand*{\mysubfigfactor}{0.465}

\DeclareMathOperator*{\argmax}{arg\,max}
\DeclareMathOperator*{\argmin}{arg\,min}
\DeclareMathOperator*{\sumsum}{\sum\sum}

\usepackage{centernot}

\newcommand{\vect}{\boldsymbol}

\newcommand{\vectab}[2]{#1_1,\ldots,#1_{#2}}
\newcommand{\vectabc}[3]{#1_{1}#3,\ldots,#1_{#2}#3}

\newcommand{\Seta}[1]{1,\ldots,#1}
\newcommand{\Set}[1]{\mathcal{#1}}

\newcommand{\vectx}{\vect{x}}
\newcommand{\vecty}{\vect{y}}

\newcommand{\texth}[1]{\!^\text{#1}}

\newcommand{\STATExx}{\STATEx\hspace{.5\algorithmicindent}}

\usepackage{forloop}
\newcounter{loopcntr}
\newcommand{\rpt}[2][1]{%
	\forloop{loopcntr}{0}{\value{loopcntr}<#1}{#2}%
}

\newcommand{\subgroup}[1]%
{\rlap{\smash{%
			\newcount\cnt%
			\cnt \numexpr#1\relax%
			\advance\cnt -1\relax%
			$\tabcolsep=.1em\begin{tabular}[t]{|l}\multicolumn{1}{l}{}\\%
			\rpt[\cnt]{\\}
			\\\hline\end{tabular}$%
		}}}

		\usetikzlibrary{decorations.pathreplacing,calc}
		\newcommand{\tikzmark}[1]{\tikz[overlay,remember picture,baseline=-1ex] \node (#1) {};}
		\newcommand{\allclust}{\vect{\overline{\Set{C}}}}

		\newcounter{myRefCount}

\begin{document}

\title{\huge Dynamic Clustering and ON/OFF Strategies for Wireless Small Cell Networks
}

\author{
	\IEEEauthorblockN{Sumudu Samarakoon,~\IEEEmembership{Student Member, IEEE}, Mehdi Bennis,~\IEEEmembership{Senior Member, IEEE}, Walid Saad,~\IEEEmembership{Senior Member, IEEE} and Matti Latva-aho,~\IEEEmembership{Senior Member, IEEE}}
	\thanks{This research was supported by the Finnish Funding Agency for Technology and Innovation (TEKES), Nokia, Anite, and Huawei Technologies, the SHARING project under the Finland grant 128010, and the U.S. National Science Foundation (NSF) under Grants CNS-1460333, CNS-1460316, and CNS-1513697.
		
		Sumudu Samarakoon, Mehdi Bennis, and Matti Latva-aho are with Department of Communications Engineering, University of Oulu, Finland (email: \{sumudu,bennis,matti.latva-aho\}@ee.oulu.fi).
		
		Walid Saad was a corresponding author, is with Wireless@VT, Bradley Department of Electrical and Computer Engineering, Virginia Tech, Blacksburg, VA´.
		He was also International Scholar at Department of Computer Engineering, Kyung Hee University, Seoul, South Korea (email: walids@vt.edu).
		
		The authors would like to thank Dr. Kimmo Hiltunen, Ericsson, Finland, for his valuable comments which helped to improve the quality of this work.}
}


\maketitle
\nopagebreak[4]

\begin{abstract}
	
	In this paper, a novel cluster-based approach for maximizing the energy efficiency of wireless small cell networks is proposed.
	A dynamic mechanism is proposed to group locally-coupled small cell base stations (SBSs) into clusters based on location and traffic load.
	Within each formed cluster, SBSs coordinate their transmission parameters to minimize a cost function which captures the tradeoffs between energy efficiency and flow level performance, while satisfying their users' quality-of-service requirements.
	Due to the lack of inter-cluster communications, clusters compete with one another in order to improve the overall network's energy efficiency.
	This inter-cluster competition is formulated as a noncooperative game between clusters that seek to minimize their respective cost functions.
	To solve this game, a distributed learning algorithm is proposed using which clusters autonomously choose their optimal transmission strategies based on local information.
	It is shown that the proposed algorithm converges to a stationary mixed-strategy distribution which constitutes an epsilon-coarse correlated equilibrium for the studied game.
	Simulation results show that the proposed approach yields significant performance gains reaching up to $36\%$ of reduced energy expenditures and up to $41\%$ of reduced fractional transfer time compared to conventional approaches.

\end{abstract}



\section{Introduction}\label{sec:introduction}

In the past decade, wireless services have evolved from traditional voice and text messaging to advanced applications such as video streaming, multimedia file sharing, and social networking~\cite{online:comscore13}.
Such bandwidth-intensive applications increase the load of existing wireless cellular systems and potentially lead to an increased energy consumption~\cite{pap:arshad12,pap:brevis11,jnl:kyuho11}.
The deployment of low-cost and high-capacity small cells over existing cellular networks has been introduced as a promising solution to offload the macro cellular traffic to small cell networks~\cite{onln:yan14}.
However, a systematic deployment of small cells may cause severe inter-cell interference and can increase the network's overall energy consumption.
Thus, it is important to analyze and address the inter-cell interference in wireless small cell networks~\cite{jnl:rengarajan11}.
In this respect, the development of energy-efficient resource management mechanisms for wireless small cell networks has become a major research challenge in recent years~\cite{pap:arshad12,pap:brevis11,jnl:kyuho11,pap:zhou11,pap:bhuaumik10,pap:celebi13,pap:soh13}.

Existing literature has studied energy efficiency enhancements of small cells under various scenarios.
In \cite{pap:arshad12}, an optimal deployment strategy is proposed for a two-tier network based on power minimization subject to a target spectral efficiency.
The work in \cite{pap:brevis11} introduces a stochastic programing approach to minimize energy consumption by optimizing small cell base station (SBS) locations.
In \cite{jnl:kyuho11}, a greedy algorithm for turning base stations (BSs) on or off is proposed to improve the tradeoff between energy efficiency and traffic delay.
A probabilistic sleep-wake mechanism is presented in~\cite{pap:zhou11} to optimize energy efficiency of relays in conventional cellular networks.
In \cite{pap:bhuaumik10}, a new approach to find the optimal cell size and sleep/wake state is studied focusing on energy consumption.
A random ON/OFF strategy for data networks is proposed in \cite{pap:celebi13} for energy reduction by analyzing the distributions of delay tolerant and intolerant users.
Although these interesting works enable notable energy reductions in small cell networks, they do not
provide thorough analytical studies on the convergence and the optimality of the solution obtained from the proposed algorithms.
In addition, some of practical problems of small cells such as the selfish behavior of BSs, the inefficient BS switching between ON and OFF, the underutilized radio spectrum as well as the potential undesirable outages at BSs are not taken into account~\cite{pap:bhuaumik10,pap:hongseok13,pap:celebi13}.

One promising solution to overcome these challenges is to enable the BSs to coordinate their transmission.
In this regard, different types of BS clustering methods~\cite{book:rokach05,jnl:luxburg07,book:jacob09,pap:ogston03,pap:backstrom11} and cluster based coordination mechanisms~\cite{jnl:amr14,pap:hoisseini12,pap:lee13} have been recently proposed.
In \cite{jnl:amr14}, a clustering-based resource allocation and interference management scheme is proposed for femtocells based on exhaustive search.
The authors in \cite{pap:hoisseini12} introduce an inter-cluster interference coordination technique to improve the sum rate of a femtocell network.
In \cite{pap:lee13}, a threshold based BS sleep algorithm is proposed to improve the energy efficiency of cellular networks by using a Delaunay triangulation graph.
Although these studies take advantage of clustering, they are based on heuristics with no convergence proofs and in which  BS clustering is static and does not capture network dynamics such as variation of BS spatio-temporal loads or user arrivals and departures.

The main contribution of this paper is to propose a new cluster-based ON/OFF strategy that allows SBSs to dynamically switch ON and OFF, depending on the current, slowly varying traffic load, user requirements, and network topology.
In small cell networks, performing dynamic approaches for switching BSs ON and OFF may require the knowledge of the entire network to operate effectively which incurs significant overhead.
Therefore, coordination mechanisms with minimum overhead are needed to group BSs into clusters within which BSs can smartly and locally coordinate their transmissions.
Unlike previous studies~\cite{jnl:amr14,pap:hoisseini12,pap:lee13}, we investigate not only location-based clustering methods, but we also consider the effects of BS capabilities to dynamically handle traffic, and further compare the performance of centralized and decentralized clustering solutions.
Moreover, we propose a novel inter- and intra-cluster implicit coordination mechanism for dynamically switching ON/OFF BSs while striking a balance between throughput and energy consumption.
In the proposed approach, each cluster of BSs aims at minimizing a cost function which captures the individual energy consumption and the flow level performance in terms of the fractional time required to transfer their traffic load.
Since clusters compete among each other, we cast the problem as a noncooperative game among clusters of BSs.
To solve this game, we propose a distributed algorithm using notions from regret learning, in which clusters implicitly coordinate their transmissions without information exchange~\cite{pap:bennis12,jnl:qian12}.
Furthermore, we prove that the proposed algorithm converges to a stationary distribution which results in an $\epsilon$-coarse correlated equilibrium in the noncooperative game, and the optimality of the solution is analyzed.
Finally, simulation results show that the proposed approach improves the energy efficiency and reduces the overall network traffic significantly as compared to conventional approaches.

The rest of the paper is organized as follows.
The system model is presented in Section~\ref{sec:system_model} and the problem formulation and cluster-based coordination are provided in Section~\ref{sec:problem_form}.
The proposed game theoretical approach for the decentralized BS ON/OFF mechanism and the convergence and optimality analysis of the solution are discussed in Section~\ref{sec:game_solution}.
Section~\ref{sec:results} provides simulation results and, finally, conclusions are drawn in Section~\ref{sec:conclusions}.

\subsubsection*{Notation}\label{sec:notation}

The regular symbols represent the scalars while the boldface symbols are used for vectors.
$\|\vectx\|$ denotes the Euclidean norm of the vector $\vectx$.
The sets are denoted by upper case calligraphic symbols.
$|\Set{X}|$ and $\Delta(\Set{X})$ represent the cardinality and the set of all probability distributions of the finite set $\Set{X}$, respectively.
The function $\mathbbm{1}_{\Set{X}}(x)$ denotes the indicator function defined as $\mathbbm{1}_{\Set{X}}(x)=1$ if $x\in\Set{X}$ and $\mathbbm{1}_{\Set{X}}(x)=0$ if $x\notin\Set{X}$,
and $x^+=\max(x,0)$ defines the selection of the largest non-negative component of $(x,0)$.
Let $\mathbb{E}_{\,\vect\pi}\big[f(\vectx)\big]$ be the expected value of the function $f(\cdot)$ of random vector $\vectx$ over a probability distribution $\vect\pi$.
Furthermore, $\mathbb{R}$, $\mathbb{Z}_+$, $\mathbf{1}$ and $\mathbf{0}$ represent the set of real numbers, the set of non-negative integers, vector of all ones, and vector of all zeros, respectively.


\section{System Model and Problem Formulation}\label{sec:system_model}

\subsection{System Model}\label{sec:network_model}

Consider the downlink of a wireless small cell network whose BSs $\Set{B}=\{\Seta{B}\}$ which includes macro BSs (MBSs) and SBSs.
We assume that BSs are uniformly distributed over a two-dimensional network layout and we let $\vectx$ be the vector of two-dimensional coordinates with respect to the origin.
Let $\Set{L}_b$ be the coverage area of BS $b$ such that any given user equipment (UE) at a given location $\vectx$ is served by BS $b$ if $\vectx\in\Set{L}_b$.

We assume that UEs connected to BS $b$ are heterogeneous in nature such that each UE has a different QoS requirement based on its individual packet arrival rate.
In this respect, let $\eta_b(\vectx)$ and $R_b(\vectx)$ be the traffic influx rate and the data rate of any UE at location $\vectx\in\Set{L}_b$ which follows an M/M/1 queue.
The fraction of time BS $b$ needs to serve the traffic ${\eta_b(\vectx)}$ from BS $b$ to location $\vectx$ is defined as,
\begin{equation}\label{eqn:load_density}
\varrho_b(\vectx)=\textstyle\frac{\eta_b(\vectx)}{R_b(\vectx)}.
\end{equation}
Consequently, the \emph{fractional transfer time} of BS $b$ is given by~\cite{jnl:hongseok12}:
\begin{equation}\label{eqn:load_bs}
\rho_b = \int_{\vectx\in\Set{L}_b} \varrho_b(\vectx) d\vectx. 
\end{equation}
Here, the fractional transfer time $\rho_b$ represents the fractional time required for a BS to deliver its requested traffic.
Moreover, the average number of flows at BS $b$ is given by $\frac{\rho_b}{1-\rho_b}$ and it is proportional to the expected delay at BS $b$~\cite[pp. 169]{book:gallager92},\cite{jnl:hongseok12}.
Thus, the parameter $\rho_b$ which indicates the flow level behavior in terms of fractional time can be referred to as \emph{time load} or simply \emph{load}, hereinafter~\cite{jnl:kyuho11,pap:hongseok13,jnl:hongseok12}.
For any BS, a successful transmission implies that the BS delivers the traffic to its respective UEs, i.e. $\rho_b\in(0,1)$ for any $b\in\Set{B}$.
Therefore, the effective transmission power $P^{\texth{Work}}_b$ of BS $b$, when it uses a transmission power of $P_b$ is $P^{\texth{Work}}_b = \rho_b P_b$.
A switched-ON BS $b$ consumes $P^{\texth{Base}}_b$ power to operate its radio frequency components and baseband unit in addition to the effective transmission power $P^{\texth{Work}}_b$~\cite{pap:georgios12B}.
From an energy efficiency perspective, some BSs might have an incentive to switch OFF.
This allows to reduce the power consumption of the baseband units by a fraction $q_b<1$ by turning OFF the radio frequency components~\cite{pap:frenger11,pap:hilthunen13}.
However, during the OFF state, BSs need to sense the UEs in their vicinity and thus, have non-zero energy consumption ensured by $q_b>0$.
Let $I_b$ be the transmission indicator of BS $b$ such that $I_b=1$ indicates the ON state while $I_b=0$ reflects the OFF state.
Thus, the power consumption of BS $b$ can be given by:
\begin{equation}\label{eqn:state_based_power}
P^{\texth{Total}}_b =
\begin{cases}
q_bP_b^{\texth{Base}} &\mbox{if}~I_b=0, \\
\big( \frac{1}{\vartheta_b}P_b^{\texth{Work}} + P_b^{\texth{Base}} \big) &\mbox{if}~I_b=1,
\end{cases}
\end{equation}
where $\vartheta_b(<1)$ captures the efficiency of the power amplifier unit and losses in the main supply and cooling units of BS $b$~\cite{pap:georgios12B}.
We note that, with this model, we are able to capture not only the power consumption due to the actual transmission, but also the variations in the power consumption which are required to maintain the BS in either ON or OFF modes.

For the downlink transmission, the channels between BSs and UEs are modeled
as additive white Gaussian noise (AWGN) channels with noise variance $N_0$~\cite{onln:3gpp10}.
We assume that all BSs transmit on the same frequency spectrum (i.e., co-channel deployment).
Thus, the data rate of a UE at location $\vectx\in\Set{L}_b$ which is served by BS $b\in\Set{B}$ is given by:
\begin{equation}\label{eqn:rate_ue}
R_b(\vectx) = \omega\log_2\Bigg(1 + \frac{ P_bh_{b}(\vectx) }{ \sum_{\forall b'\in\Set{B}\setminus b} P^{\texth{Work}}_{b'}I_{b'}h_{b'}(\vectx) + N_0 } \Bigg),
\end{equation}
where $\omega$ is the bandwidth and $h_{b'}(\vectx)$ is the channel gain between BS $b'$ and UE at $\vectx$.
We assume that each BS uses orthogonal resource blocks in the time domain to serve its UEs, and there is no central controller available in the system to coordinate the transmission power and RB allocation among BSs.

\subsection{Problem Formulation}\label{sec:problem_form}

Our objective is to improve the energy saving of the network while minimizing the time loads of all the BSs during their downlink transmission.
The downlink transmission consists of three phases:
1) all the UEs determine the BSs to which they wish to connect,
2) the decentralized transmission power and ON/OFF state selection of BSs in each cluster is performed, and
3) the data transmission and acquiring the knowledge on the network based on  limited information -- cluster-based \emph{distributed learning} -- take place.

\subsubsection{User association}\label{sec:ue_association_policy}

Classical UE association mechanisms are often based on the maximum received signal strength (RSS) or signal to interference and noise ratio (SINR)~\cite{pap:richter12}.
However, these criteria are oblivious to the time load (and thus, the expected number of flows and delays) which may lead to overloading BSs yielding lower spectral efficiencies.
Thus, a smarter mechanism in which BSs advertise their current load to all UEs within their coverage area is needed~\cite{jnl:kyuho11,jnl:hongseok12}.

At time instant $t$, each BS $b$ advertises its estimated load $\hat{\rho}_b(t)$ via a broadcast control message.
Considering both the received signal strength and load at time $t$, UE at location $\vectx$ selects  BS $b(\vectx,t)$, where $\vectx\in\Set{L}_{b(\vectx, t)}$, as follows:
%
\begin{equation}\label{eqn:ue_association}
b(\vectx, t) = \argmax_{b \in \Set{B}} \Big\{ \Big(1-\hat{\rho}_b(t)\Big)^{n}P_b^{\texth{Rx}}(t) \Big\}.
\end{equation}
Here, $P_b^{\texth{Rx}}(t) =  P_b(t)I_b(t) h_{b}(\vectx,t)$
is the received signal power at UE in location $\vectx$ from BS $b$ at time $t$.
The impact of the load is determined by the coefficient $n\geq 0$.
The classical reference signal strength indicator (RSSI)-based UE association is a special case of (11) where $n = 0$.
For $n>0$, the BS load affects the UE association alongside the received signal power.
According to (\ref{eqn:ue_association}), from the UE perspective, each UE selects a BS assuming it is the most suitable one.
However, the BS might have a preference of switching OFF by offloading its traffic to one of its neighbor.
Thus, the actual serving BS might be different from the BS chosen by the UE.
Hereinafter, we refer to UE's choice as the \emph{anchor BS}.

Since the BSs need to estimate their load beforehand, the estimation must accurately reflect the actual load.
Here, we use the notion of moving time-average load estimation $\hat{\rho}_b(t)=\int_{\vectx\in\Set{L}_b} \hat{\rho}_b(\vectx,t) d\vectx$ at time $t$ based on history as follows:
%
\begin{equation}\label{eqn:load_estimation}
\hat{\rho}_b(\vectx,t) = \nu(t) \rho_b(\vectx,t-1) + \Big( 1 - \nu(t) \Big) \hat{\rho}_b(\vectx,t-1),
\end{equation}
where $\hat{\rho}_b(\vectx,t)$ is the moving time-average load estimation at location $\vectx$, $\nu(t)$ is the learning rate of the load estimation\footnote{
	Choosing the learning rate such that $\nu(t)\in[0,1]$ for all $t\in\mathbb{N}$ and $\lim_{t\to\infty} \nu(t) = 0$ ensures that $\hat{\rho}_b(t)$ yields the time average load of BS $b$ as $t\to\infty$.
} and the initial condition as a predefined preferred load, i.e. $\hat{\rho}_b(\vectx,0)=\rho^{\texth{Pref}}$.
Leveraging different time-scales, we assume that $\nu(t)$ is selected such that the load estimation procedure~(\ref{eqn:load_estimation}) is much slower than the UE association process given in (\ref{eqn:ue_association}).

\subsubsection{Decentralized transmit power and ON/OFF state selection}\label{subsec:decent_power}

Once the UE association is done, the BSs need to optimize their transmissions.
The configuration of the entire network is defined by the transmit powers and the ON/OFF states of all BSs.
This configuration is captured by a transmission power vector $\vect{P}=(\vectab{P}{|\Set{B}|})$
and an ON/OFF state indicator vector $\vect{I}=(\vectab{I}{|\Set{B}|})$.
Therefore, we use the tuple $(\vect{P},\vect{I})$ to represent the network configuration.

For a given network configuration $(\vect{P},\vect{I})$, BS $b$ consumes $P_b^{\texth{Total}}$ amount of energy and handles a load $\rho_b$ by serving a set of UEs in its coverage area $\Set{L}_b$.
Based on (\ref{eqn:load_bs})-(\ref{eqn:rate_ue}), it can be seen that $\frac{\partial \rho_b}{\partial P_b}<0$.
Therefore, from the perspective of uncoordinated BSs that are oblivious to the choices made by the rest of the network, each individual BS will conclude that a load reduction is viable only with an increased energy consumption.
Hence, there is a tradeoff between load and energy consumption reduction.
Here, for each BS $b\in\Set{B}$, we define a cost function that captures both energy consumption and load, as follows:
\begin{equation}\label{eqn:bs_utility_metric}
\Upsilon_b(\vect{P},\vect{I}) = \lambda \underbrace{\frac{P_b^{\texth{Total}}}{\big( \frac{1}{\vartheta_b} P_b^{\texth{Max}} + P_b^{\texth{Base}}\big)}}_{\text{energy consumption}} \; + \; \mu \underbrace{
	\vphantom{\frac{P_b^{\texth{Total}}}{\big( \frac{1}{\vartheta_b} P_b^{\texth{Max}} + P_b^{\texth{Base}}\big)}}
	\rho_b}_{\text{load}},
\end{equation}
where the coefficients $\lambda$ and $\mu$ are weight parameters that indicate the impact of energy and load on the cost, and $P_b^{\texth{Max}}$ is the maximum transmit power of BS $b$.
The objective is to minimize the total network-wide cost $\sum_{\forall b\in\Set{B}}\Upsilon_b(\vect{P},\vect{I})$.

Whenever a given BS switches OFF, it should maintain service to its UEs via a re-association process.
For a proper re-association, BSs need to coordinate with the rest of the network which requires a centralized controller.
Such a centralized approach involves large information exchange within the network incurring large delays.
Endowing individual BSs or sets of well-chosen and locally-coupled BSs with decision making capabilities is crucial.
Therefore, clustering is leveraged to group locally-coupled BSs in terms of mutual interference and  load together.
Allowing coordination per cluster enables efficient UE load balancing with limited information exchange.
Furthermore, intra-cluster coordination is carried out for joint load balancing and interference management.

In view of the above, we consider that all BSs are grouped into sets of judiciously-chosen clusters $\allclust = \{\vectab{\Set{C}}{|\overline{\vect{\Set{C}}}|}\}$.
Here, a given cluster $\Set{C}_i\in\allclust$ consists of a set of BSs which can coordinate with one another.
The clustering mechanism is discussed in Section \ref{sec:cluster_solution} in details.
For each cluster $\Set{C}_i\in\allclust$, the per-cluster cost $\Upsilon_{\Set{C}_i}$ is the aggregated cost of each cluster member, $\Upsilon_{\Set{C}_i}(\vect{P},\vect{I}) = \sum_{\forall b\in\Set{C}_i} \Upsilon_b(\vect{P},\vect{I})$.
Thus, the minimization of the network cost is given by the following optimization problem:
%
\begin{subequations}\label{eqn:optimization_energy_efficiency_network}
	\begin{eqnarray}\label{eqn:objective_function}
	\underset{\vect{P},\vect{I},\allclust}{\text{minimize}} && \textstyle\sum_{\forall\Set{C}_i\in\allclust} \Upsilon_{\Set{C}_i}(\vect{P},\vect{I}), \\ 
	\label{cns:BI_belongs_to_cluster} \text{subject to} && |\Set{C}_i|\geq 1, \quad \forall \Set{C}_i\in\overline{\vect{\Set{C}}}, \\
	\label{cns:exclusive_clusters} && \Set{C}_i\cap\Set{C}_j=\emptyset, \quad \forall \Set{C}_i,\Set{C}_j\in\overline{\vect{\Set{C}}},~\Set{C}_i\neq\Set{C}_j, \\
	\label{cns:tx_params} &&
	\rho_b \!\in\! [0,1],
	P_b \!\leq \!P_b^{\texth{Max}}\!\!\!,
	I_b \!\in \!\{0,1\}, \forall b\in\Set{B}.
	\end{eqnarray}
\end{subequations}
Here, constraints (\ref{cns:BI_belongs_to_cluster}) and (\ref{cns:exclusive_clusters}) ensure that any BS is part of only one cluster.
Solving (\ref{eqn:optimization_energy_efficiency_network}) involves BS clustering, UE association, and transmission power optimization.

\subsubsection{Data transmission and inter-cluster distributed learning}

Due to the absence of a central controller, the clusters need to determine the network configuration $(\vect{P},\vect{I})$ based on their limited knowledge.
Using the knowledge of the cluster members, clusters need to implicitly coordinate and learn their transmission powers and ON/OFF states, and find the best network configuration which solves (\ref{eqn:optimization_energy_efficiency_network}).
As shown in Section \ref{sec:game_solution}, the cost minimization problem can be posed as a noncooperative game that allows to minimize the costs of each cluster.


\section{Cluster Formation and Intra-Cluster Coordination}\label{sec:cluster_solution}

Solving (\ref{eqn:optimization_energy_efficiency_network}) requires an efficient BS clustering and intra-cluster coordination mechanism.
For clustering, we map the system into a weighted graph $G=(\Set{B},\Set{E})$.
Here, the set $\Set{B}$ represents the BSs while the set $\Set{E}$ represents the links between locally-coupled BSs.

\subsection{Similarity Matrix-Based Clustering}\label{subsec:similarities}

The key step in clustering is to identify similarities between BSs to group BSs with similar characteristics.
This will allow to perform coordination between BSs with little signaling overhead.
While many similarity features can be used, two important ones include the locations of BSs and their loads.
The BS locations define the capability of coordination among nearby BSs while the BS loads reflect mutual interference and the willingness of cooperation.
Next, we examine a number of techniques to calculate similarities between BSs based on location and load.

\paragraph{Neighborhood based on BS adjacency and static Gaussian similarity}\label{sec:loaction_similarity}

Consider the graph $G=(\Set{B},\Set{E})$.
The set of edges $\Set{E}$ indicates the adjacency between nodes.
We use a parameter $e_{bb'}$ to characterize the presence of a link between nodes $b$ and $b'$ when $e_{bb'}=1$.
Let $\vecty_b$ be the coordinates of the vertex (or BS) $b\in\Set B$ in the Euclidean space.
The neighborhood of node $b$ is defined by the adjacency of the nodes in its $\varepsilon_d$-neighborhood:
\begin{equation}\label{eqn:neighborhood}
\Set{N}_{b}=\{b'|~e_{bb'}=e_{b'b}=1, \|\vecty_b-\vecty_{b'}\|\leq\varepsilon_d \}.
\end{equation}

The links between the vertices are weighted based on their similarities.
The \emph{radial basis function kernel} known as the Gaussian similarity, is a popular kernel function used for location based classifications in machine learning~\cite{pap:yuille98,lec:shashua09,pap:sung11}.
Thus, the Gaussian similarity is based on the distance between BSs $b$ and $b'$ calculated as follows:
\begin{equation}\label{eqn:similarity_distance}
s^{d}_{bb'} =
\begin{cases}
\exp\big(\frac{-\|\vecty_b-\vecty_{b'}\|^2}{2\sigma_d^2}\big) & \mbox{if}~\|\vecty_b-\vecty_{b'}\|\leq\varepsilon_d,\\
0 & \mbox{otherwise},
\end{cases}
\end{equation}
where $\sigma_d$ controls the impact of neighborhood size\footnote{
	For a given $\varepsilon_d$, the range of the distance-based similarity for any two connected BSs is $[e^{-\varepsilon_d/2\sigma_d^2},1]$. The lower bound is determined by the selection of $\sigma_d$.
}.
Furthermore, the distance-based similarity matrix $\vect{S}^d$ is formed using $s^d_{bb'}$ as the $(b,b')$-th entry.
The rationale behind (\ref{eqn:similarity_distance}) is that the BSs located far from each others have low similarities, and as they come closer, similarities increase in which BSs are more likely to cooperate with each other.

\paragraph{Load-based dissimilarity}

Unlike the static distance-based clustering in (\ref{eqn:similarity_distance}), load-based clustering provides a more dynamic  manner of grouping neighboring BSs in terms of time load.
Consider BSs $b$ and $b'$ with the loads $\rho_b$ and $\rho_{b'}$.
The BSs with similar loads yield no benefit by offloading traffic to one another, and thus, there is no advantage for such BSs to cooperate when $\rho_b\approx\rho_{b'}$.
On the other hand, if $\rho_b\gg\rho_{b'}$, BS $b'$ has the capability of handling additional traffic and thus, is capable to cooperate with $b$.
Based on these observations, the weight of the link between $b$ and $b'$ should be \emph{(i)} positive and \emph{(ii)} increasing with the load difference.
Therefore, a dissimilarity metric is used for the time load instead of similarity.
Thus, a modified version of (\ref{eqn:similarity_distance}) is used to calculate load-based similarity between BSs $b$ and $b'$ as follows:
\begin{equation}\label{eqn:similarity_load}
s^l_{bb'} = \exp\bigg(\frac{\|\rho_b-\rho_{b'}\|^2}{2\sigma_l^2}\bigg),
\end{equation}
where $\sigma_l$ controls the range of the similarity\footnote{
	The range of the load-based similarity is $[1,e^{1/2\sigma_l^2}]$. The upper bound of the dissimilarity depends on the choice of $\sigma_l$.
}.
The value of $s^l_{bb'}$ is used as the $(b,b')$-th entry of the load-based similarity matrix $\vect{S}^l$.

\paragraph{Combining the similarities}

One approach for combining $\vect{S}^d$ and $\vect{S}^l$ is by forming a linear combination such as $\big(\theta \vect{S}^d + (1-\theta)\vect{S}^l\big)$ with $0\leq\theta\leq 1$.
According to (\ref{eqn:neighborhood}), any two BSs $b,b'$ with $e_{bb'}=0$ will not be able to perform any coordination, and thus, the similarity between BSs $b$ and $b'$ needs to be equal to zero.
Since $s_{bb'}^l>0$ as per (\ref{eqn:similarity_load}), a linear combination of $s_{bb'}^d$ and $s_{bb'}^l$ with $\theta\in[0,1)$ always results in a positive similarity between BSs $b$ and $b'$ which contradicts the previous claim.
Therefore, the classical linear combination of $\vect{S}^d$ and $\vect{S}^b$ cannot be used for clustering.

In order to preserve the characteristics of similarities, the joint similarity $\vect S$ with $s_{bb'}$ as the $(b,b')$-th element is formulated as follows:
\begin{equation}\label{eqn:similarity_joint}
s_{bb'} = (s_{bb'}^d)^\theta \cdot (s_{bb'}^l)^{(1-\theta)},
\end{equation}
where $0\leq\theta\leq 1$ controls the impact of the distance and the load similarities on the joint similarity.
Here, cooperation is possible only if a physical link between nodes exists\footnote{
	We assume that $0^\theta=0$ is held for $\theta=0$ as well.
}, i.e. $\forall \theta\in[0,1]$, $e_{bb'}=0\implies s_{bb'}=0$.

\begin{theorem}\label{thm:gauss_sim}
	The joint similarity $\vect S$ formulated in (\ref{eqn:similarity_joint}) simplifies to a Gaussian similarity, and thus, all the properties of the Gaussian similarity hold for $\vect S$.
\end{theorem}
\begin{proof}
	Let $s_{bb'} = (s_{bb'}^d)^\theta \cdot (s_{bb'}^l)^{(1-\theta)}$ is the $(b,b')$-th element of $\vect S$ with  $\theta\in[0,1]$, $s^{d}_{bb'} = \exp\big(\frac{-\|\vecty_b-\vecty_{b'}\|^2}{2\sigma_d^2}\big)$, $s^{l}_{bb'} = \exp\big(\frac{\|\rho_b-\rho_{b'}\|^2}{2\sigma_l^2}\big)$, and $\vecty_i = (y_i^{(1)},y_i^{(2)})$ for $i=\{b,b'\}$.
	Then,
	\begin{align}
	\nonumber s_{bb'} &= (s_{bb'}^d)^\theta\cdot(s_{bb'}^l)^{(1-\theta)} \\
	&= \exp\bigg(\frac{-(y_{b}^{(1)}-y_{b'}^{(1)})^2-(y_{b}^{(2)}-y_{b'}^{(2)})^2}{2(\sigma_d^2/\theta)} + \frac{(\rho_{a}-\rho_{b})^2}{2(\sigma_l^2/1-\theta)}\bigg).
	\end{align}
	Define a mapping $\vect{\zeta}_i=(\zeta_i^{(1)},\zeta_i^{(2)},\zeta_i^{(3)})$ for $i=\{b,b'\}$ with,
	\begin{align}\label{eqn:mapping_coordinates}
	\zeta_i^{(1)} = \frac{y_i^{(1)}\sqrt{\theta}}{\sigma_d}\sigma, && \zeta_i^{(2)} = \frac{y_i^{(2)}\sqrt{\theta}}{\sigma_d}\sigma, && \zeta_i^{(3)} = \frac{\imath\rho_i\sqrt{1-\theta}}{\sigma_l}\sigma,
	\end{align}
	where $\imath=\sqrt{-1}$ and $\sigma$ is an arbitrary constant.
	Using the above mapping,
	\begin{align}
	\nonumber s_{bb'}
	= \exp\bigg(\frac{-\|\vect{\zeta}_{b}-\vect{\zeta}_{b'}\|^2}{2\sigma^2}\bigg).
	\end{align}
	Thus, $s_{bb'}$ follows the relation of Gaussian similarity and thus, $\vect S$ satisfies the properties of Gaussian similarity matrix.
\end{proof}

Note that as $\vect S$ follows the form of a Gaussian similarity, the location and the load of any BS can be mapped into a 3-dimensional (3D) coordinates.
Therefore, any clustering method which uses the joint similarity of distance and load can be interpreted as a distance-based clustering in a 3D space.

\subsection{Clustering}\label{sec:clustering_techniques}

Once the similarities $\vect S$ of the graph $G=(\Set{B},\Set{E})$ are formed, BS clustering is performed.
There are many clustering mechanisms available in the literature~\cite{book:rokach05,jnl:luxburg07,book:jacob09,pap:ogston03,pap:backstrom11}.
Given the knowledge of the entire similarity matrix, any of these clustering methods can be applied.
However, the main challenge is the need to obtain the full similarity matrix of the entire network.
In essence, clustering methods can be categorized into two subgroups: centralized and decentralized clustering.
Centralized clustering requires the knowledge of the entire similarity matrix, based on which, a central controller performs the clustering process.
Decentralized clustering methods depend on the information gathered from neighbors and decision making is done per node.

With this in mind, we focus on three clustering methods: \emph{i)} $k$-mean clustering, \emph{ii)} spectral clustering, and \emph{iii)} peer-to-peer searched-based clustering.
The main intention is to identify the performance gains due to the capability of exploiting nodes' physical locations (spectral clustering over $k$-mean clustering), and the capability of the decentralization (peer-to-peer search-based clustering over spectral clustering and $k$-mean clustering).
We assume that each method yields a set of clusters $\vect{\overline{\Set{C}}}=\{\vectab{\Set C}{|\vect{\overline{\Set{C}}}|}\}$ where cluster ${\Set C}_i\in\vect{\overline{\Set{C}}}$ consists of a set of BSs who  cooperate with each other.

\subsubsection{$k$-mean clustering~\cite{book:rokach05}}

The objective of $k$-mean clustering is to partition the set of nodes into $|\vect{\overline{\Set{C}}}|=k$ clusters in which each node belongs to the cluster with the nearest mean distance.
This is commonly used for nodes distributed in a physical space (an area or a volume).
Formally speaking, $k$-mean clustering is given by,
\begin{eqnarray}\label{eqn:k-mean_clusters}
\underset{\allclust}{\argmin} && \textstyle \sum_{\forall\Set{C}_i\in\allclust} \sum_{\forall b\in\Set{C}_i} \|\vect{\zeta}_b-\overline{\vect{\zeta}}_{\Set{C}_i}\|,
\end{eqnarray}
where $\vect{\zeta}_b$ represents the coordinates of node $b\in\Set{B}$ and $\overline{\vect{\zeta}}_{\Set{C}_i}=\frac{1}{|\Set{C}_i|}\sum_{\forall b\in\Set{C}_i}\vect{\zeta}_b$ is the center of the cluster $\Set{C}_i$.
As per Theorem \ref{thm:gauss_sim}, using the mapping between the location $\vecty_b$ and the load $\rho_b$ of BS $b$  into a 3D vector $\vect{\zeta}_b$, solving (\ref{eqn:k-mean_clusters}) yields BS clusters based on joint similarity.
However, this clustering algorithm requires the knowledge of all the locations and loads of BSs and thus, it is categorized as a centralized clustering method.

\subsubsection{Spectral clustering~\cite{jnl:luxburg07}}

Unlike $k$-mean clustering, spectral clustering method exploits the connectivity and the compactness, thus, the geometry of the node distribution in a graph.
The graph Laplacian matrix is formed as $\vect{L}=\vect{D}-\vect{S}$ where $\vect{D}$ is the diagonal matrix with the $b$-th diagonal element given as $\sum_{b'=1}^{|{\Set{B}}|} s_{bb'}$.
For a given cluster size $|\allclust|$, the concatenation of the smallest $|\allclust|$ eigenvectors of $\vect{L}$ can be used to evaluate the modified $k$-dimensional coordinates $\vectab{\vect{\psi}}{|\Set{B}|}$.
Then, $k$-mean clustering is applied on $\vectab{\vect{\psi}}{|\Set{B}|}$ to produce the clusters $\allclust$.

The number of clusters $k$ is 
closely related to the eigenvalues of $\vect L$ and is given by~\cite{jnl:luxburg07};
\begin{equation}\label{eqn:spectral_clust_number}
k = \underset{i}{\argmax} \big(|\varsigma_{i+1} - \varsigma_i|\big), \quad i=\Seta{|\Set B|},
\end{equation}
where $\varsigma_i$ is the $i$-th smallest eigenvalue of $\vect L$.
When nodes are distributed as $k$ groups, the first $k$ eigenvalues become small while the $(k+1)$-th eigenvalue becomes relatively large.
As the node distribution follows more likely a uniform distribution, the difference between consecutive eigenvalues becomes constant.
In such cases, when the area of interest is large, nodes will be isolated from each other and thus, $k$ will become large.
If the area is small, all the nodes are close to each other, and thus, will be grouped into a single cluster, i.e. $k=1$.
Due to fact that the spectral clustering algorithm requires full knowledge of $\vect S$, it is also a centralized clustering mechanism.

\subsubsection{Peer-to-peer search-based clustering}

The location based decentralized clustering method for a multi-agent network is presented in~\cite{pap:ogston03}.
Here, each agent in the system searches other similar agents in a peer-to-peer (P2P) fashion.
The knowledge over the neighborhood is sufficient to form the clusters, and the size of any cluster is limited by a predefined parameter ($|\Set{C}|^{\texth{Max}}$), i.e. $|\Set{C}_i|\leq|\Set{C}|^{\texth{Max}}$ for any $\Set{C}_i\in{|\allclust|}$.
Modifying the parameters of the algorithm in~\cite{pap:ogston03} including both location and load information, we leverage the P2P-search based (P2P-SB) clustering technique for BSs.

\subsection{Coordination}\label{subsec:coordination}

Clustering BSs serves three main purposes: {\it (i)} reducing the signaling overhead required for coordination among BSs compared to a system with a centralized controller, {\it (ii)} reducing intra-cluster interference among locally coupled BSs, and {\it (iii)} efficiently offloading UEs to ON-BSs from BSs which need to switch OFF.
As intra-cluster coordination fulfills the above requirements, the number of switched-OFF BSs is expected to increase compared to the case without clusters and coordination.

Once clusters are formed, the lightly loaded BS within the cluster is selected as the cluster head\footnote{
	Partitioning the macrocell area and selecting the BSs located at the center of each partition and assigning them as cluster heads is another option.
}.
The function of a cluster head is to coordinate the transmissions between the cluster members.
Consequently, the entire traffic load of the cluster is distributed between its members in which orthogonal resource allocation helps to mitigate intra-cluster interference.
As the BSs within a cluster have the ability to coordinate, the entire cluster can be seen as a \emph{single super BS} serving all UEs within its vicinity as illustrated in Fig.~\ref{fig:coordination}.

\begin{figure*}[t]
	\centering
	\includegraphics[width=.88\linewidth]{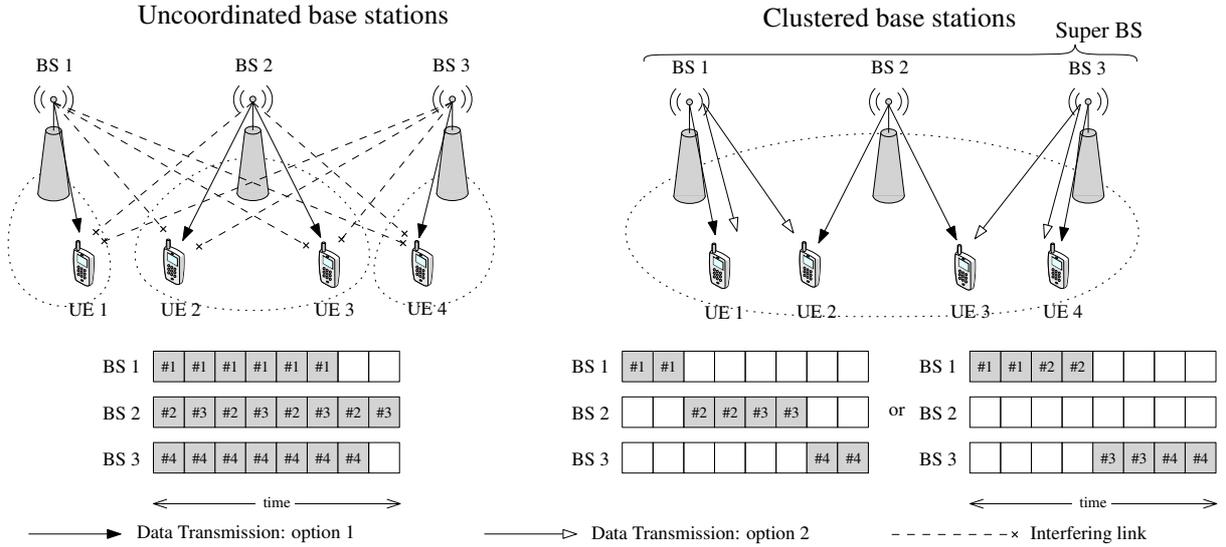}
	\caption{Intra-cell interference mitigation through clustering based coordination. Clustered base stations do not interfere with each other and users may offloaded.}
	\label{fig:coordination}
\end{figure*}

By enabling such re-association of UEs between BSs within a given cluster, BSs can efficiently be switched OFF while ensuring users' QoS.
The actual load of each BS is unknown before the downlink transmission takes place.
Therefore, UE offloading within the cluster is carried out with the objective of minimizing the estimated load in a given cluster.
This can eventually reduce the number of UEs served with low rates.
Let $\Set{M}_{\Set{C}_i}$ be the set of UEs whose anchor BSs belong to cluster $\Set{C}_i$.
Here, we focus on a discrete coverage area by redefining $\Set{L}_b$ as the collection of circular regions with negligible radii around the UEs associated with BS $b$.
Thus, the integration forms of $\rho_b$ and $\hat{\rho}_b$ can be modeled as a summation over the set of UEs.
Let $z_{bm}^{\Set{C}_i}\in\{0,1\}$ be an indicator which defines the connectivity between UE $m\in\Set{M}_{\Set{C}_i}$ and BS $b\in\Set{C}_i$, such that $z_{bm}^{\Set{C}_i}=1$ if UE $m$ is served by BS $b$ and, $z_{bm}^{\Set{C}_i}=0$ otherwise\footnote{
	Note that $z_{bm}^{\Set{C}_i}=1\Leftrightarrow\vectx_m\in\Set{L}_b$ and $z_{bm}^{\Set{C}_i}=0\Leftrightarrow\vectx_m\notin\Set{L}_b$. Therefore, the BS load $\rho_b$ is a linear combination of $z_{bm}^{\Set{C}_i}, \forall m\in\Set{M}_{\Set{C}_i}$ and $\vect{\varrho}_b$ as in (\ref{eqn:load_bs}).
}.
Thus, the scheduling problem within cluster $\Set{C}_i\in\allclust$ can be written as follows:
\begin{subequations}\label{eqn:scheduling_MI_opt}
	\begin{eqnarray}
	\underset{\vect{z}^{\Set{C}_i}}{\text{minimize}} && \textstyle\sum_{\forall b\in\Set{C}_i}  \sum_{\forall m\in\Set{M}_{\Set{C}_i}} \hat{\rho}_{bm}z_{bm}^{\Set{C}_i} \\
	\label{cns:unique_connection} \text{subject to} && \textstyle\sum_{\forall b\in\Set{C}_i} z_{bm}^{\Set{C}_i} = 1, \quad \forall m\in\Set{M}_{\Set{C}_i} \\
	\label{cns:boolean} && z_{bm}^{\Set{C}_i} \in \{0,1\}, \quad \forall m\in\Set{M}_{\Set{C}_i}, \forall b \in \Set{C}_i,
	\end{eqnarray}
\end{subequations}
where $\vect{z}^{\Set{C}_i}$ consists of all $z_{bm}^{\Set{C}_i}$ indicators for all UEs $m\in\Set{M}_{\Set{C}_i}$ and $b \in \Set{C}_i$.
The constraints ensure that any UE is served by a single cluster member.
This is a \emph{mixed-integer linear program} (MILP) and thus, an NP-hard problem.
Nevertheless, relaxing the integer constraint (\ref{cns:boolean}) yields a \emph{linear program} which can be easily solved.
Therefore, the relaxed problem is given by:
\begin{subequations}\label{eqn:scheduling_relaxed}
	\begin{eqnarray}
	\underset{\hat{\vect{z}}^{\Set{C}_i}}{\text{minimize}} && \textstyle\sum_{\forall b\in\Set{C}_i} \sum_{\forall m\in\Set{M}_{\Set{C}_i}} \hat{\rho}_{bm} \hat{z}_{bm}^{\Set{C}_i} \\
	\label{cns:unique_connection_new} \text{subject to} && \textstyle\sum_{\forall b\in\Set{C}_i} \hat{z}_{bm}^{\Set{C}_i} = 1, \quad \forall m\in\Set{M}_{\Set{C}_i} \\
	\label{cns:boolean_relaxed} && 0 \leq \hat{z}_{bm}^{\Set{C}_i} \leq 1, \quad \forall m\in\Set{M}_{\Set{C}_i}, \forall b \in \Set{C}_i.
	\end{eqnarray}
\end{subequations}
where $\hat{\vect{z}}^{\Set{C}_i}$ is the relaxed optimization variable of $\vect{z}^{\Set{C}_i}$.
A suboptimal solution for (\ref{eqn:scheduling_MI_opt}), $({\vect{z}}^{\Set{C}_i})^\star$, is obtained as follows:
\begin{equation}\label{eqn:schedule_solution}
(z_{bm}^{\Set{C}_i})^\star =
\begin{cases}
1 & \mbox{if~} (\hat{z}_{bm}^{\Set{C}_i})^\star = \argmax_{\forall b' \in \Set{C}_i}\Big( (\hat{z}_{b'm}^{\Set{C}_i})^\star \Big), \\
0 & \mbox{otherwise}.
\end{cases}
\end{equation}

\subsection{Modeling the Overhead Costs Due to Intra-cluster Coordination}\label{subsec:overhead_cost}

The signaling overhead due to intra-cluster coordination is a significant factor for a fair comparison between clustered and non-clustered approaches.
However, most of the existing works \cite{pap:hoisseini12,jnl:amr14,pap:bandyopadhyay03,jnl:baracca14}, and \cite{jnl:katranaras14} simply ignore the intra-cluster signaling overhead while some studies such as  \cite{jnl:hammoudeh08} and \cite{pap:band03} point out that signaling overhead depends on the cluster size in terms of both the number of members and the area within which the cluster is located.
However, those works do not provide any concrete model of overhead.
Motivated by the above facts, we model the signaling overhead cost due to intra-cluster coordination as an increment of power consumption $\delta P_b^{Base}$ and it is calculated by,
\begin{equation}\label{eqn:overheadcost_model}
\delta P_b^{\texth{Base}} = \chi(|\Set{N}_b|-1)\varepsilon_d,
\end{equation}
where the parameter $\chi$ defines the sensitivity of the power increment to neighborhood size and range.
Clearly, the isolated BSs (when either $\varepsilon_d=0$ or $|\Set{N}_b|=1$ with $\varepsilon_d>0$) do not suffer from signaling overhead due to intra-cluster coordination.
For a clustered BS, the total overhead cost is the summation of overheads due to the individual communications with each member of the cluster.
Assuming all the BSs in a cluster are treated equally yields an equal individual overhead and, thus, the total overhead will be varying \emph{linearly} with the size of the neighborhood.
As the range of neighborhood $\varepsilon_d$ increases, there is a higher chance to increase the cluster size and, consequently, the signaling overhead.
Since $\varepsilon_d$ and $|\Set{N}_b|$ are coupled with one another, we simply assume that the signaling overhead due to communication with another generic BS has a linear relation with the range of the neighborhood.
Thus, a change of the range $\varepsilon_d$ alone has a linear impact on $\delta P_b^{\texth{Base}}$ while the changes of $|\Set{N}_b|$ due to $\varepsilon_d$ will further impact the signaling overhead $\delta P_b^{\texth{Base}}$.
Moreover, the above choice allows us to explicitly include the cost of intra-cluster coordination in our original objective and thus, optimize the transmission parameters accordingly.


\section{Self-Organizing Inter-Cluster Switching ON/OFF Mechanism}\label{sec:game_solution}

Given the formation of clusters, our next goal is to develop a self-organizing mechanism for solving (\ref{eqn:optimization_energy_efficiency_network}) in which each cluster of BSs adjusts its transmission parameters based on local information.
To do so, we use a regret-based learning approach~\cite{jnl:perlaza13,pap:bennis12}.

\subsection{Game Formulation}\label{subsec:game_formulation}

In the proposed approach, clusters need to autonomously select their transmission configurations to minimize their cost functions.
However, the cell coverage and the achievable throughput of BSs depend not only on the action of their own cluster, but also on the choices of neighboring clusters due to interference.
In this regard, we formulate a noncooperative game $\Set{G} = \big( \allclust, \{\Set{A}_{\Set{C}_i}\}_{\Set{C}_i\in\allclust}, \{{u}_{\Set{C}_i}\}_{\Set{C}_i\in\allclust} \big)$ in which the set of clusters $\allclust$ is the set of players.
Each player $\Set{C}_i\in\vect{\overline{\Set{C}}}$ has a set $\Set{A}_{\Set{C}_i}=\big\{ a_{\Set{C}_i,1},\ldots,a_{\Set{C}_i,|\Set{A}_{\Set{C}_i}|} \big\}$ of actions where an action of cluster $\Set{C}_i$, $a_{\Set{C}_i}$, is composed of the configurations of all its cluster members, i.e. $a_{\Set{C}_i}\triangleq(P_{b_1},I_{b_1},\ldots,P_{b_{|\Set{C}_i|}},I_{b_{|\Set{C}_i|}})$ with $\Set{C}_i=\{b_1,\ldots,b_{|\Set{C}_i|}\}$.
The utility function of cluster $\Set{C}_i$ is $u_{\Set{C}_i}(a_{\Set{C}_i},\vect{a}_{-\Set{C}}) = -\Upsilon_{\Set{C}_i}(\vect{P},\vect{I})$  with $u_{\Set{C}_i}:(\Set{A}_{\Set{C}_i},\Set{A}_{-\Set{C}_i})\mapsto\mathbb{R}$ where $a_{\Set{C}_i}$ is the action of cluster $\Set{C}_i$ and $\vect{a}_{-\Set{C}_i}\in\Set{A}_{-\Set{C}_i}$ are the actions of the other clusters.

Let $\vect{\pi}_{\Set{C}_i}(t) = \big( \pi_{\Set{C}_i,1}(t),\ldots,\pi_{\Set{C}_i,|\Set{A}_{\Set{C}_i}|}(t) \big)$ be a probability distribution using which cluster $\Set{C}_i$ selects a given action from $\Set{A}_{\Set{C}_i}$ at time instant $t$.
Here, $\pi_{\Set{C}_i,j}(t)=\mbox{Pr}\big(a_{\Set{C}_i}(t)=a_{\Set{C}_i,j}\big)$ is the cluster $\Set{C}_i$'s \emph{mixed strategy} where $a_{\Set{C}_i}(t)$ is the action of player $\Set{C}_i$ at time $t$.
When cluster $\Set{C}_i$ plays action $a_{\Set{C}_i}(t)$, it observes its utility feedback\footnote{
	Since the action played by the cluster $\Set{C}_i$ at time $t$ is already defined as $a_{\Set{C}_i}(t)$, we use ${u}_{\Set{C}_i}(t)$ instead of ${u}_{\Set{C}_i}\big(t;a_{\Set{C}_i}(t)\big)$ for the sake of notation simplicity.
	The system utility ${u}\big(t;\vect{a}(t)\big)=\sum_{\forall\Set{C}_i\in\allclust} {u}_{\Set{C}_i}\big(t;a_{\Set{C}_i}(t)\big)$ is simply written as ${u}(t)=\sum_{\forall\Set{C}_i\in\allclust} {u}_{\Set{C}_i}(t)$.
}  ${u}_{\Set{C}_i}\big(t;a_{\Set{C}_i}(t)\big)$ based on which it minimizes its regret associated to action $a_{\Set{C}_i}(t)$.
Subsequently, player $\Set{C}_i$ estimates its utility ${\vect{\hat{u}}}_{\Set{C}_i}(t)=\big( \hat{u}_{\Set{C}_i,1}(t),\ldots,\hat{u}_{\Set{C}_i,|\Set{A}_{\Set{C}_i}|}(t) \big)$ and regret ${\vect{\hat{r}}}_{\Set{C}_i}(t)=\big( \hat{r}_{\Set{C}_i,1}(t),\ldots,\hat{r}_{\Set{C}_i,|\Set{A}_{\Set{C}_i}|}(t) \big)$ for each action assuming it has played the same action during all previous time slots $\Seta{t-1}$.
At each time $t$, player $\Set{C}_i$ updates its mixed strategy probability distribution $\vect{\pi}_{\Set{C}_i}$ in which the actions with higher regrets are exploited while exploring the actions with low regrets. 
Such behavior is captured by the Boltzmann-Gibbs (BG) distribution $\vect{G}_{\Set{C}_i}=(G_{\Set{C}_i,1},\ldots,G_{\Set{C}_i,|\Set{A}_{\Set{C}_i}|})$ with $\vect{G}_{\Set{C}_i}:{\vect{\hat{r}}}\mapsto\Delta(\Set{A}_{\Set{C}_i})$ which is calculated as follows:
%
\begin{equation}\label{eqn:BG_distribution}
G_{\Set{C}_i,j}\big(\vect{\hat{r}}_{\Set{C}_i}(t)\big) = \frac{\exp\big(\kappa \hat{r}_{\Set{C}_i,j}^+(t)\big)} {\sum_{\forall j'\in\Set{A}_{\Set{C}_i}} \exp\big(\kappa \hat{r}_{\Set{C}_i,j'}^+(t)\big) }, \: j\in\Set{A}_{\Set{C}_i},
\end{equation}
where $\kappa>0$ is a temperature parameter which balances between exploration and exploitation.
At each time $t$, all the estimations for any player $\Set{C}_i\in\allclust$, $\vect{\hat{u}}_{\Set{C}_i}(t),~\vect{\hat{r}}_{\Set{C}_i}(t)$ and $\vect{\pi}_{\Set{C}_i}(t)$, are updated as follows;
\begin{equation}\label{eqn:algoUpdates}
\begin{cases}
{\hat{u}}_{\Set{C},i}(t) &= {\hat{u}}_{\Set{C},i}(t-1)  \\
&\hfill + \tau_{\Set{C}_i}(t)\mathds{1}_{\{a_{\Set{C},i}=v_{\Set{C}_i}(t-1)\}} \Bigl({u}_{\Set{C}_i}(t)-{\hat{u}}_{\Set{C},i}(t-1)\Bigr),\\
{\hat{r}}_{\Set{C},i}(t) &= {\hat{r}}_{\Set{C},i}(t-1)  \\
&\hfill +\iota_{\Set{C}_i}(t) \Bigl({\hat{u}}_{\Set{C},i}(t-1)-{u}_{\Set{C}_i}(t-1)-{\hat{r}}_{\Set{C},i}(t-1)\Bigr),\\
\pi_{\Set{C},i}(t) &= \pi_{\Set{C},i}(t-1) \\
&\hfill + \varepsilon_{\Set{C}_i}(t) \Bigl(G_{\Set{C},i}\big({\vect{\hat{r}}}_{\Set{C}_i}(t-1)\big)-\pi_{\Set{C},i}(t-1)\Bigr).
\end{cases}
\end{equation}
with the learning rates satisfying $\lim_{t\to\infty}\frac{\tau(t)}{\iota(t)}=\lim_{t\to\infty}\frac{\tau(t)}{\varepsilon(t)}=0$, and  $\lim_{t\to\infty}\sum_{n=1}^t\xi(n) = +\infty$ and $\lim_{t\to\infty}\sum_{n=1}^t\xi^2(n) < +\infty$ for all $\xi=\{\tau,\iota,\varepsilon\}$.
Our choice of the learning rates follows the format of ${1}/{t^\phi}$ with exponent $\phi\in(0.5,1)$.

\begin{figure*}[!t]
	\centering
	\includegraphics[width=\textwidth]{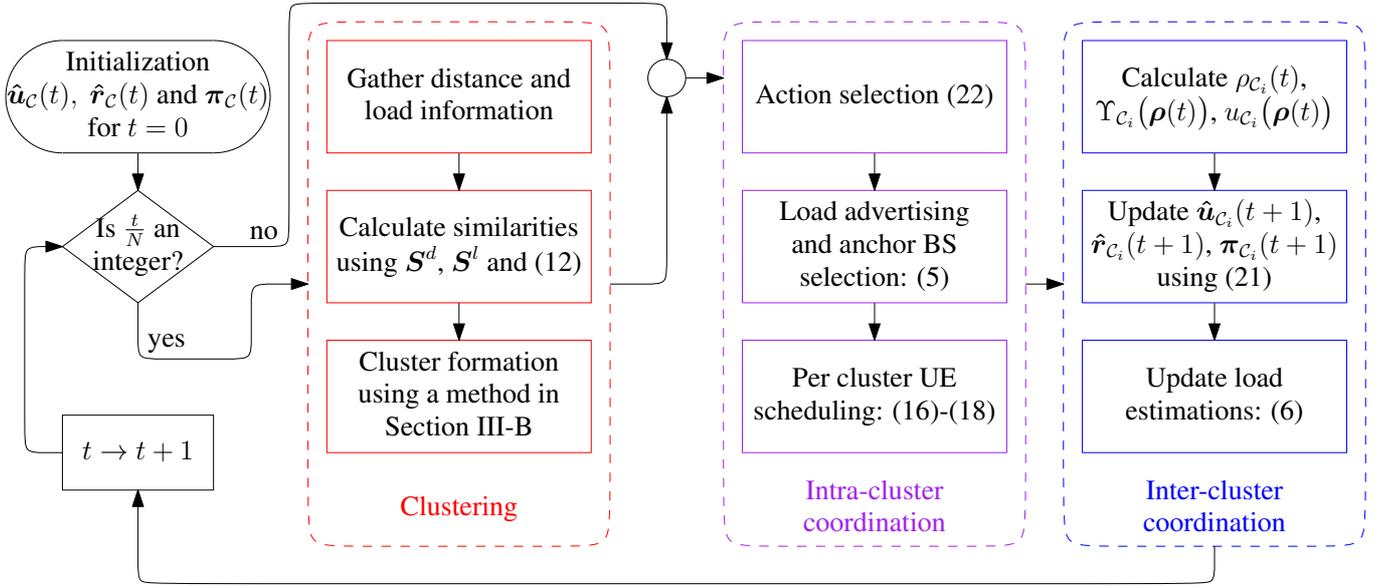}
	\caption{Flow diagram representation of the proposed method.}
	\label{fig:flow}
\end{figure*}

To solve the game, we propose a novel learning algorithm that proceeds as follows.
At the beginning of each time instant $t$, the clusters select their actions $(a_{\Set{C}_i})$ based on their mixed-strategy probabilities $(\pi_{\Set{C}_i})$, i.e. for all ${\Set{C}_i}\in\allclust$,
\begin{equation}\label{eqn:prob_to_action}
a_{\Set{C}_i}(t) = f\big( \pi_{\Set{C}_i}(t-1) \big),
\end{equation}
where $f:\pi_{\Set{C}_i}\mapsto a_{\Set{C},i}\in\Set{A}_{\Set{C}_i}$ is the mapping from probability distribution to an action.
Depending on the actions, all BSs advertise their loads $(\hat{\rho}_b)$ and UEs select their anchor BSs.
Based on the estimated loads and actions, UE scheduling within clusters takes place.
All the clusters of BSs carry out the transmission based on the actions $\vect{a}(t)=\big(a_1(t),\ldots,a_{|\overline{\Set{C}}|}(t)\big)$ and calculate the utilities $u_{\Set{C}_i}(t)$ for all ${\Set{C}}_i\in\allclust$.
Each cluster individually updates its utility and regret estimations $\big(\vect{\hat{u}}_{\Set{C}_i}(t),\vect{\hat{r}}_{\Set{C}_i}(t)\big)$ along with the mixed strategy probabilities $\big(\vect{\pi}_{\Set{C}_i}(t)\big)$ while all the BSs update their load estimations $\vect{\hat{\rho}}(t)$.
With the updated load estimations, the clusters are updated for each time interval $N$.
The entire operation of the proposed method is illustrated in Fig.~\ref{fig:flow} and the proposed algorithm is summarized in Algorithm \ref{alg:active_deactive}.

Next, we invoke the Gibbs-Markov equivalence to show that the system following (\ref{eqn:BG_distribution})-(\ref{eqn:prob_to_action}) is a \emph{Gibbs Field} with steady-state distribution.
\begin{theorem}\label{thm:convergence}
	Let $\vect{a}(t)=\big(\vectabc{a}{|\allclust|}{(t)}\big)\in\Set{A}$ be the collection of actions played by all the clusters based on their mixed strategy probabilities, where $\Set{A}=\prod_{\Set{C}_i\in\allclust}\Set{A}_{\Set{C}_i}$.
	Let $\vect{\pi}(t)=\big(\vectabc{\pi}{|\Set{A}|}{(t)}\big)$ with $\pi_j(t)=\mbox{Pr}\big(\vect{a}(t)=\vect{a}_{j}\big), \forall \vect{a}_j\in\Set A$, be the action selection probability of the system.
	Under Algorithm \ref{alg:active_deactive}, as $t\to +\infty$, $\vect{\pi}(t)$ converges to a \emph{stationary distribution} $\vect{\Pi}=(\Pi_{\vect{a}},\forall\vect{a}\in\Set{A})$ with,
	\begin{equation}\label{eqn:stationary_distribution}
	\Pi_{\vect a} = \frac{\exp(\kappa \hat{\Gamma}_{\vect a})}{\sum_{\forall\vect{a'}\in\Set{A}} \exp(\kappa \hat{\Gamma}_{\vect a'})},
	\end{equation}
	where $\hat{\Gamma}_{\vect a}$ is the ensemble average of the action $\vect a$'s regret estimation as $t\to +\infty$.
\end{theorem}
\begin{IEEEproof}
	Let $U_{\Set{C}_i}$ and $U=\sum_{\forall\Set{C}_i\in\vect{\overline{\Set{C}}}}U_{\Set{C}_i}$ be the ensemble averages of the utilities of cluster $\Set{C}_i$ and the system, respectively.
	As $t\to\infty$, for the cluster $\Set{C}_i\in\overline{\vect{\Set{C}}}$, the ensemble averages of utility estimation, regret estimation and mixed-strategy probability become $\vect{\hat{U}}_{\Set{C}_i}, \vect{\hat{\Gamma}}_{\Set{C}_i}=\vect{\hat{U}}_{\Set{C}_i}-U_{\Set{C}_i}\mathbf{1}$, and $\vect{\Pi}_{\Set{C}_i,j}=\vect{G}_{\Set{C}_i,j}( \vect{\hat{\Gamma}}_{\Set{C}_i})$ for $j\in\Set{A}_{\Set{C}_i}$ as per (\ref{eqn:algoUpdates}), respectively.
	Since $\sum_{\forall\Set{C}_i\in\overline{\vect{\Set{C}}}} {\hat{\Gamma}}_{\Set{C}_i,{a}_{\Set{C}_i}} = \sum_{\forall\Set{C}_i\in\overline{\vect{\Set{C}}}} ( \hat{U}_{\Set{C}_i, a_{\Set{C}_i}}-U_{\Set{C}_i} ) = \hat{U}_{\vect{a}}-U = \hat{\Gamma}_{\vect{a}}$ for any action $\vect{a}\in\Set{A}$, $\vect{\hat{\Gamma}} = (\hat{\Gamma}_{\vect a}, \forall\vect{a}\in\Set{A})$ is the ensemble average of the regret estimations for the entire system.
	Due to the fact that an action $\vect a$ of the system is composed by the set of actions ($a_{\Set{C}_i,j_i},\forall \Set{C}_i\in\allclust$) per each cluster, the probability of selecting $\vect a$ is $\prod_{\forall\Set{C}_i\in\allclust}\vect{\Pi}_{\Set{C}_i,j_i}= \prod_{\forall\Set{C}_i\in\allclust}\vect{G}_{\Set{C}_i,j_i}( \vect{\hat{\Gamma}}_{\Set{C}_i})$ which can be simplified to $\Pi_{\vect a}$ in (\ref{eqn:stationary_distribution}).
	
	Following the notation in \cite[Chapter 7]{book:bremaud91}, we define a Gibbs field with a set of cliques $\allclust$, configuration space of $\Set{V}=\prod_{\forall \Set{C}_i\in\allclust}\Set{V}_{\Set{C}_i}$ with a finite size, set of finite potentials $\{\hat{\vect{\Gamma}}_{\Set{C}_i}(v_i),\forall v_i\in\Set{V}_{\Set{C}_i}\}_{\forall \Set{C}_i\in\allclust}$, and probability distribution $\vect{\Pi}=(\Pi_{v},\forall v\in\Set{V})$ where $\Pi_{v} = \frac{\exp\big(\kappa \hat{\Gamma}({v})\big)}{\sum_{\forall v'\in\Set{V}} \exp\big(\kappa \hat{\Gamma}({v'})\big)}$ and $\hat{\Gamma}({v})=\sum_{\forall\Set{C}_i\in\allclust}\hat{\vect{\Gamma}}_{\Set{C}_i}(v_i)$ with $v=(\vectab{v}{|\allclust|})$.
	Suppose the configuration $v(t)$ at time $t$ changes to the configuration $v(t+1)$ at time $t+1$ according to $\vect{\Pi}(t)$.
	Along with the properties we have introduced, the evolution of the configurations in the above Gibbs field is equivalent to a Markov chain~\cite[Chapter 7, Theorem 2.1]{book:bremaud91}.
	Since the configuration space $\Set{V}$ is finite at any given time $t$, the evolution of configurations follows a time-inhomogeneous Markov chain.
	Note that $\vect{\Pi}(t)\succ \mathbf{0}$ due to the fact that the potentials are finite.
	Henceforth, all configurations have self-loops with positive probability and thus, the Markov chain is aperiodic.
	Furthermore, $\vect{\Pi}(t)$ being a positive vector implies that $\mbox{Pr}\big(v(t+1)=v''|v(t)=v'\big)>0$ for any $v',v''\in\Set V$.
	Thus, it verifies that the process can start from configuration $v'$ and ends in configuration $v''$ with a positive probability in which the irreducible and positive recurrence properties are held.
	Since the time-inhomogeneous Markov chain is aperiodic, irreducible and positive recurrent, as $t\to\infty$ its transition probability $\vect{\Pi}$ converges to a stationary distribution~\cite[Chapter 3, Theorem 3.1]{book:bremaud91}.
	
	It is shown that the transition probability of the above Gibbs field is equivalent to the action selection probability provided by Algorithm \ref{alg:active_deactive}.
	Since the transition probability converges to a stationary distribution, we can claim that $\vect{\Pi}$ given in (\ref{eqn:stationary_distribution}) becomes a stationary distribution as well.
\end{IEEEproof}

\begin{algorithm}[!t]
	\caption{Dynamic Clustering and BS Switching ON/OFF}
	\label{alg:active_deactive}
	\begin{algorithmic}[1]                    
		\STATE {\bf Input:} $\vect{\hat{u}}_{\Set{C}}(t),~\vect{\hat{r}}_{\Set{C}}(t)$ and $\vect{\pi}_{\Set{C}}(t)$ for $t=0$ and $\forall {\Set{C}}\in\allclust$, and $\vect{\hat{\rho}}(t)$ \tikzmark{anchor}
		\WHILE{ true }
		\STATE $t\rightarrow t+1$
		\STATExx\subgroup{3}{\emph{{Intra-cluster operations}:}}
		\STATE Action selection: $a_{\Set{C}_i}(t) = f\Big( \vect{\pi}_{\Set{C}_i}(t-1) \Big)$,~(\ref{eqn:prob_to_action})
		\STATE Load advertising $\vect{\hat{\rho}}(t)$ and anchor BS selection: (\ref{eqn:ue_association})
		\STATE Per cluster UE scheduling: (\ref{eqn:scheduling_MI_opt})-(\ref{eqn:schedule_solution})
		\STATExx\subgroup{3}{\emph{{Inter-cluster operations}:}}
		\STATE Calculations: $ \rho_{\Set{C}_i}(t),~\Upsilon_{\Set{C}_i}\big(\vect{\rho}(t)\big),~u_{\Set{C}_i}\big(\vect{\rho}(t)\big) $
		\STATE Update utility and regret estimations, and probability:
		\STATE \hspace{10pt} $\vect{\hat{u}}_{\Set{C}_i}(t+1),~\vect{\hat{r}}_{\Set{C}_i}(t+1),~\vect{\pi}_{\Set{C}_i}(t+1)$,~(\ref{eqn:algoUpdates})
		\STATE Update load estimations: (\ref{eqn:load_estimation})
		\IF{ $\frac{t}{N}\in\mathbb{Z}_+$}
		\STATE Update clusters $\allclust$.
		\ENDIF
		\ENDWHILE
	\end{algorithmic}
	
\end{algorithm}

Note that the convergence to a stationary distribution is affected by the dynamically changing clusters.
Dynamic changes in clustering are based on load estimation (locations are fixed) and mixed strategy probabilities depend on the utility and regret estimations.
Therefore, minimizing the fluctuations in estimations helps to achieve the stationary distribution.
This is achieved by
\emph{i)} defining the cost as a separable function,
\emph{ii)} avoiding fully connected graph with properly selected $\varepsilon_d, \sigma_d$ and $\sigma_l$ ensures physically-separated clusters and thus, prevents radical changes to the mixed strategy probabilities, and
\emph{iii)} introducing time-scale separation for clustering and power optimization.
The separable cost function allows the cost calculation to be unchanged through the entire network.
Further, we restrict clusters to remain unchanged for a time interval of $N(\gg 1)$ allowing (\ref{eqn:algoUpdates}) to yield stable estimations.
The changes in clusters diminish as the load estimations become stable.
Henceforth, the system yields a stationary distribution and form stable clusters.

\subsection{Optimality of The Solution}\label{subsec:optimality}

Theorem \ref{thm:convergence} shows that for a fixed $\kappa$, the mixed strategy probabilities converge to the stationary distribution $\vect{\Pi}$.
In this section we analyze the effect of BG temperature parameter $\kappa$ on the optimality of the solution at the convergence of Algorithm~\ref{alg:active_deactive}.
Thus, we use $\vect{\Pi}^{(\kappa)}$ in order to denote the dependency of $\kappa$ on $\vect{\Pi}$.
First, we show that as $\kappa$ approaches infinity, the solution of Algorithm~\ref{alg:active_deactive} converges to the solution of (\ref{eqn:optimization_energy_efficiency_network}).

\begin{theorem}\label{thm:kappa_infinity}
	When $\kappa\to\infty$, $\Pi^{(\kappa)}_{\vect a}$ becomes
	\begin{equation}\label{eqn:kappa_infinity}
	\lim_{\kappa\to\infty} \Pi^{(\kappa)}_{\vect a} = \Pi^{(\infty)}_{\vect a}=
	\begin{cases}
	\frac{1}{|\Set{A}^\star|} & \mbox{if}~\vect{a}\in\Set{A}^\star, \\
	0 & \mbox{if}~\vect{a}\notin\Set{A}^\star,
	\end{cases}
	\end{equation}
	where $\Set{A}^\star$ is the set of global optimal solutions to (\ref{eqn:optimization_energy_efficiency_network}).
\end{theorem}
\begin{IEEEproof}
	Let $\vect a\in\Set{A}^\star$.
	Therefore, $\hat{U}_{\vect a}>\hat{U}_{\vect a'} \implies \hat{\Gamma}_{\vect a}>\hat{\Gamma}_{\vect a'}$ for all $\vect a'\notin\Set{A}^\star$.
	Thus,
	\begin{align*}
	\Pi_{\vect a}^{(\infty)} &= \lim_{\kappa\to\infty} \frac{\exp(\kappa \hat{\Gamma}_{\vect a})}{\sum_{\forall\vect{a'}\in\Set{A}} \exp(\kappa \hat{\Gamma}_{\vect a'})} \\
	&= \lim_{\kappa\to\infty} \frac{1}{|\Set{A}^\star| + \sum_{\forall\vect{a'}\notin\Set{A}^\star} \exp\big(\kappa (\hat{\Gamma}_{\vect a'}-\hat{\Gamma}_{\vect a}) \big)} 
	= \frac{1}{|\Set{A}^\star|}.
	\end{align*}
\end{IEEEproof}

Theorem \ref{thm:kappa_infinity} states that for any finite $\kappa$, since $\Pi_{\vect a'}\geq 0$ for all $\vect{a'}\notin\Set{A}^\star$, a non-optimal action can be selected with non-zero probability.
Therefore, for finite $\kappa$, we cannot ensure that Algorithm~\ref{alg:active_deactive} leads to pick optimal actions under the stationary distribution.
However, the optimality of the solution increases with $\kappa$ as given in the following theorem.

\begin{theorem}\label{thm:kappa_monotonic}
	For any optimal action $\vect a\in\Set{A}^\star$, the probability of selecting the optimal solution with $\Pi_{\vect a}^{(\kappa)}$, monotonically increases with $\kappa$.
\end{theorem}
\begin{IEEEproof}
	Let us study the first derivative of (\ref{eqn:stationary_distribution}):
	\begin{align*}
	\frac{\partial\Pi_{\vect a}^{(\kappa)}}{\partial\kappa} &= \frac{\partial}{\partial\kappa} \frac{\exp(\kappa \hat{\Gamma}_{\vect a})}{\sum_{\forall\vect{a'}\in\Set{A}} \exp(\kappa \hat{\Gamma}_{\vect a'})} = \Pi_{\vect a} \Big( \hat{\Gamma}_{\vect a} - \mathbb{E}_{\,\vect\Pi^{(\kappa)}} \big[\hat{\Gamma}_{\vect a}\big] \Big),
	\end{align*}
	where $\vect{\hat{\Gamma}} = (\hat{\Gamma}_{\vect a}, \forall\vect{a}\in\Set{A})$ is the regret estimations for the system.
	Since $\hat{\Gamma}_{\vect a}\geq\hat{\Gamma}_{\vect a'}$ for any $\vect a\in\Set{A}^\star$ and for all $\vect a'\in\Set{A}$, $ \hat{\Gamma}_{\vect a} > \mathbb{E}_{\,\vect\Pi^{(\kappa)}} \big[\hat{\Gamma}_{\vect a}\big]$ is true and thus, $\frac{\partial\Pi_{\vect a}^{(\kappa)}}{\partial\kappa}>0$ is held, i.e. the probability of choosing the optimal action monotonically increases with $\kappa$.
\end{IEEEproof}

From Theorem \ref{thm:kappa_infinity} and Theorem \ref{thm:kappa_monotonic}, we can observe that the choice of $\kappa$ affects the optimality as the solution of Algorithm~\ref{alg:active_deactive} converges to the stationary distribution.
Finally, we prove that selecting large $\kappa$ not only increases probability of choosing the optimal action, but also increases the expected utility of the system as given in the theorem below.

\begin{figure*}
	\begin{align}
	\nonumber \text{LHS}
	\nonumber  &= \textstyle\sumsum\limits_{\forall\alpha\in\Set{A}_{\Set{C}},\forall\beta\in\Set{A}_{-\Set{C}}} U_{\Set{C},(\alpha',\beta)} \Pi_{(\alpha,\beta)}^{(\infty)} - \sumsum\limits_{\forall\alpha\in\Set{A}_{\Set{C}},\forall\beta\in\Set{A}_{-\Set{C}}} U_{\Set{C},(\alpha,\beta)} \Pi_{(\alpha,\beta)}^{(\infty)} \\
	\nonumber  &\geq \textstyle\sumsum\limits_{\forall\alpha\in\Set{A}_{\Set{C}},\forall\beta\in\Set{A}_{-\Set{C}}} U_{\Set{C},(\alpha',\beta)} \Pi_{(\alpha,\beta)}^{(\kappa)} - \sumsum\limits_{\forall\alpha\in\Set{A}_{\Set{C}},\forall\beta\in\Set{A}_{-\Set{C}}} U_{\Set{C},(\alpha,\beta)} \Pi_{(\alpha,\beta)}^{(\infty)} \\
	\label{eqnlong}  &\geq \textstyle\sumsum\limits_{\forall\alpha\in\Set{A}_{\Set{C}},\forall\beta\in\Set{A}_{-\Set{C}}} U_{\Set{C},(\alpha',\beta)} \Pi_{(\alpha,\beta)}^{(\kappa)} - \sumsum\limits_{\forall\alpha\in\Set{A}_{\Set{C}},\forall\beta\in\Set{A}_{-\Set{C}}} U_{\Set{C},(\alpha,\beta)} \Pi_{(\alpha,\beta)}^{(\kappa)} - \epsilon.
	\end{align}
	\rule{\textwidth}{.5pt}
\end{figure*}

\begin{theorem}\label{thm:kappa_expected_utility}
	Once the solution of Algorithm~\ref{alg:active_deactive} converges to $\vect\Pi^{(\kappa)}$, the expected value of the system utility $\mathbb{E}_{\,\vect{\Pi}^{(\kappa)}} \big[ \hat{U}_{\vect a} \big]$ monotonically increases with $\kappa$.
\end{theorem}
\begin{IEEEproof}
	Consider the first derivative of the expected regrets:
	\begin{align*}
	\frac{ \partial\mathbb{E}_{\,\vect{\Pi}^{(\kappa)}} \big[ \hat{\Gamma}_{\vect a} \big] }{\partial\kappa}
	&= \frac{\partial}{\partial\kappa} \bigg( \textstyle\sum_{\forall\vect a\in\Set A} \Pi_{\vect a}^{(\kappa)}\hat{\Gamma}_{\vect a} \bigg) \\
	&= \sum_{\forall\vect a\in\Set A} \Pi_{\vect a}^{(\kappa)}\hat{\Gamma}_{\vect a} \big( \hat{\Gamma}_{\vect a} - \mathbb{E}_{\,\vect\Pi^{(\kappa)}} \big[\hat{\Gamma}_{\vect a}\big] \big) \\
	&= \mathbb{E}_{\,\vect\Pi^{(\kappa)}} \big[(\hat{\Gamma}_{\vect a})^2\big] - \mathbb{E}_{\,\vect\Pi^{(\kappa)}}^{\,2} \big[\hat{\Gamma}_{\vect a}\big] \\
	&= \mathbb{E}_{\,\vect\Pi^{(\kappa)}} \big[ \big( \hat{\Gamma}_{\vect a} - \mathbb{E}_{\,\vect\Pi^{(\kappa)}} \big[\hat{\Gamma}_{\vect a}\big] \big)^2 \big].
	\end{align*}
	Here, $\mathbb{E}_{\,\vect\Pi^{(\kappa)}} \big[ \big( \hat{\Gamma}_{\vect a} - \mathbb{E}_{\,\vect\Pi^{(\kappa)}} \big[\hat{\Gamma}_{\vect a}\big] \big)^2 \big]$ is the variance of $\vect{\hat{\Gamma}}$ over $\vect\Pi^{(\kappa)}$ and thus, $\frac{ \partial }{\partial\kappa} \Big( \mathbb{E}_{\,\vect{\Pi}^{(\kappa)}} \big[ \hat{\Gamma}_{\vect a} \big] \Big) > 0$ is held.
	Since $\vect{\hat{\Gamma}}$ is monotonically increasing with $\vect{\hat{U}}$ due to linear dependency, $\frac{ \partial }{\partial\kappa} \Big( \mathbb{E}_{\,\vect{\Pi}^{(\kappa)}} \big[ \hat{U}_{\vect a} \big] \Big) > 0$ is held as well.
\end{IEEEproof}

Based on the above theorems, the choice of a large $\kappa$ ensures the close global optimality.
However, further increasing $\kappa$ always improves the result and thus, $\kappa$ needs to be bounded above for practical implementations.
Following Theorem \ref{thm:kappa_expected_utility}, Corollary \ref{cor:kappa_bound} states that we can determine $\kappa$ which satisfies a given threshold.
\begin{corollary}\label{cor:kappa_bound}
	For any given threshold $\bar{U}<\mathbb{E}_{\,\vect{\Pi}^{(\infty)}} \big[ \hat{U}_{\vect a} \big]$, there exists a finite $\kappa$ such that $\bar{U}< \mathbb{E}_{\,\vect{\Pi}^{(\kappa)}} \big[ \hat{U}_{\vect a} \big] < \mathbb{E}_{\,\vect{\Pi}^{(\infty)}} \big[ \hat{U}_{\vect a} \big]$.
\end{corollary}
This result follows directly from Theorem~\ref{thm:kappa_expected_utility}.

\subsection{Convergence to An Equilibrium}\label{sec:convergence}

The results of the stationary distribution, the optimality of the solution, and its dependency on the BG temperature $\kappa$ are used to study the equilibrium of the game.
The stationary distribution $\vect \Pi$ defines the sequence of actions which the players are going to select.
If players have no intention to deviate from such a sequence based on their interest on improving the utilities, the game is considered to be in an \emph{$\epsilon$-coarse correlated equilibrium} (CCE)~\cite{jnl:perlaza13,pap:bennis12}.
\begin{definition}\label{def:epsionCCE}
	(\emph{$\epsilon$-coarse correlated equilibrium}):
	The mixed strategy probability $\vect{\Pi}=(\Pi_{\vect a},\allowbreak\forall\vect{a}\in\Set{A})$ is an $\epsilon$-CCE if, $\forall \Set{C}_i\in\allclust$ and $\forall a'_{\Set{C}_i}\in\Set{A}_{\Set{C}_i}$, where:
	\begin{equation*}\label{eqn:epsilonCCE}
	\textstyle\sum\limits_{\forall \vect{a}_{-\Set{C}_i}\in\Set{A}_{-\Set{C}_i}} U_{\Set{C}_i,(a'_{\Set{C}_i},\vect{a}_{-\Set{C}_i})}\Pi_{-\Set{C}_i,\vect{a}_{-\Set{C}_i}} - \sum\limits_{\forall \vect{a}\in\Set{A}} U_{\Set{C}_i,\vect{a}} \Pi_{\vect{a}}  \leq \epsilon,
	\end{equation*}
	for some $\epsilon>0$ with $\Pi_{-\Set{C}_i,\vect{a}_{-\Set{C}_i}}=\sum_{\forall {a}_{\Set{C}_i}\in\Set{A}_{\Set{C}_i}} \Pi_{({a}_{\Set{C}_i},\vect{a}_{-\Set{C}_i})}$ being the marginal probability distribution with respect to the action ${a}_{\Set{C}_i}$.
\end{definition}
Here, the players do not deviate from the mixed-strategy as long as the increment of the payoff is below $\epsilon$.
Finally, the convergence of the proposed algorithm to an $\epsilon$-CCE is given by the following theorem:
\begin{theorem}\label{thm:converge_to_eCCE}
	Algorithm~\ref{alg:active_deactive} converges to a stationary mixed strategy probability distribution $\vect{\Pi}^{(\kappa)}$ which constitutes an $\epsilon$-CCE for the game $\Set{G} = \big( \allclust, \{\Set{A}_{\Set{C}_i}\}_{\Set{C}_i\in\allclust},\allowbreak \{{u}_{\Set{C}_i}\}_{\Set{C}_i\in\allclust} \big)$.
\end{theorem}
\begin{proof}
	For the notation simplicity any action $\vect{a}\in\Set{A}$ is denoted as $\vect{a}\triangleq(\alpha,\beta)$ with $\alpha\in\Set{A}_{\Set{C}_i}$ and $\beta\in\Set{A}_{-\Set{C}_i}$ for any given $\Set{C}_i\in\allclust$.
	Using a simplified notation, the left hand side (LHS) of the $\epsilon$-CCE condition can be reformulated as follows:
	\begin{align*}
	\text{LHS}
	&= \textstyle\sum\limits_{\forall \beta\in\Set{A}_{-\Set{C}_i}} U_{\Set{C}_i,(\alpha',\beta)} \sum\limits_{\forall \alpha\in\Set{A}_{\Set{C}_i}} \Pi_{(\alpha,\beta)} - \\ &\qquad \qquad  \qquad\textstyle\sumsum\limits_{\forall\alpha\in\Set{A}_{\Set{C}_i},\forall\beta\in\Set{A}_{-\Set{C}_i}} U_{\Set{C}_i,(\alpha,\beta)} \Pi_{(\alpha,\beta)} \\
	&= \textstyle\sumsum\limits_{\forall\alpha\in\Set{A}_{\Set{C}_i},\forall\beta\in\Set{A}_{-\Set{C}_i}}  \Pi_{(\alpha,\beta)} \big( U_{\Set{C}_i,(\alpha',\beta)} - U_{\Set{C}_i,(\alpha,\beta)} \big).
	\end{align*}

	First, consider the scenario where $\kappa\to\infty$.
	Therefore, the mixed strategy probability of action $(\alpha,\beta)\in\Set{A}$, ${\Pi}_{(\alpha,\beta)}={\Pi}^{(\infty)}_{(\alpha,\beta)}$ is either $0$ or $1/|\Set{A}^\star|$ according to Theorem \ref{thm:kappa_infinity} and thus, LHS $=\sumsum_{(\alpha,\beta)\in\Set{A}^\star}  \big( U_{\Set{C}_i,(\alpha',\beta)} - U_{\Set{C}_i,(\alpha,\beta)} \big) / |\Set{A}^\star|$.
	Suppose $\vect{\Pi}^{(\infty)}$ does not ensure an $\epsilon$-CCE with $\epsilon=0$, i.e. LHS$\:>0$.
	This yields that $U_{\Set{C}_i,(\alpha',\beta)} > U_{\Set{C}_i,(\alpha'',\beta)}$ for some $(\alpha'',\beta)\in\Set{A}^\star$ and $\alpha'\in\Set{A}_{\Set{C}_i}$.
	Since the regret ${\Gamma}_{\Set{C}_i,(\alpha',\beta)}$ monotonically increases with ${U}_{\Set{C}_i,(\alpha',\beta)}$, the result is a higher regret for action $(\alpha',\beta)$ over the actions $(\alpha'',\beta)\in\Set{A}^\star$, i.e. ${\Gamma}_{\Set{C}_i,(\alpha',\beta)}>{\Gamma}_{\Set{C}_i,(\alpha'',\beta)}$.
	However, according to Theorem \ref{thm:kappa_infinity}, it can occur only when $(\alpha',\beta)\in\Set{A}^\star$ and $(\alpha'',\beta)\notin\Set{A}^\star$.
	This contradicts the former definition of the optimal set of solutions.
	Therefore, the assumption LHS$\:>0$ with $\vect{\Pi}^{(\infty)}$ does not hold, and $\vect{\Pi}^{(\infty)}$ converges to an $\epsilon$-CCE with $\epsilon=0$.

	Theorem \ref{thm:kappa_expected_utility} states that the expected utility with any finite $\kappa$ is lower than with $\kappa\to\infty$.
	Suppose the difference between expected utilities for a given finite and infinite $\kappa$ values is no greater than a positive value $\epsilon$ in which $\sum_{\forall(\alpha,\beta)\in\Set{A}} U_{\Set{C},(\alpha,\beta)} \Pi_{(\alpha,\beta)}^{(\infty)} - \sum_{\forall(\alpha,\beta)\in\Set{A}} U_{\Set{C},(\alpha,\beta)} \Pi_{(\alpha,\beta)}^{(\kappa)} \leq \epsilon$.
	Furthermore, the maximum expected utility with $\vect{\Pi}^{(\infty)}$ ensures that $\sum_{\forall(\alpha,\beta)\in\Set{A}} U_{\Set{C},(\alpha',\beta)} \Pi_{(\alpha,\beta)}^{(\kappa)} \leq \sum_{\forall(\alpha,\beta)\in\Set{A}} U_{\Set{C},(\alpha,\beta)} \Pi_{(\alpha,\beta)}^{(\infty)}$ for any $\alpha'\in\Set{A}_{\Set{C}}$ with $\vect{\Pi}^{(\kappa)}$.
	Combining above results the LHS is remodeled as (\ref{eqnlong}).
	Since LHS$\:\leq 0$ is held as discussed before, the result is $\text{LHS} + \epsilon \leq \epsilon$ and thus, $\vect{\Pi}^{(\kappa)}$ becomes an $\epsilon$-CCE as per Definition~\ref{def:epsionCCE}.
\end{proof}

\subsection{Practical aspects of the choice of $\kappa$}\label{subsec:kappa_for_converg}

Based on the above discussion regarding the optimality and convergence, it can be seen that the choice of $\kappa$ should be large to ensure optimality.
However, note that the optimality is only relevant when the algorithm converges to the steady state and the utilities and regrets are calculated in terms of their ensemble averages.
Before achieving convergence, the learning algorithm needs sufficient time to explore all the actions in order to make accurate estimations on the ensemble averages.

Consider a large $\kappa$ and the initial step $t=1$ with an arbitrary action $j_0\in\Set{A}_{\Set{C}_i}$ that has been played by the player ${\Set{C}_i}$.
Since no other action has been played so far except $j_0$, the estimations of the utility and the regret for action $j_0$ will be dominant, i.e. $\hat{u}_{\Set{C}_i,j_0}(1)>\hat{u}_{\Set{C}_i,j}(1)$ and $\hat{r}_{\Set{C}_i,j_0}(1)>\hat{r}_{\Set{C}_i,j}(1)$ for all $j\in\Set{A}_{\Set{C}_i}\setminus\{j_0\}$.
According to (\ref{eqn:BG_distribution}), as a result of using a large $\kappa$, $\exp\big(\kappa\hat{r}_{\Set{C}_i,j_0}(1)\big)\gg\exp\big(\kappa\hat{r}_{\Set{C}_i,j}(1)\big)$ results in $G_{\Set{C}_i,j_0}\big(\vect{\hat{r}}_{\Set{C}_i}(1)\big)\approx 1$ while $G_{\Set{C}_i,j}\big(\vect{\hat{r}}_{\Set{C}_i}(1)\big)\approx 0$ for all $j\in\Set{A}_{\Set{C}_i}\setminus\{j_0\}$.
Since the learning rate $\varepsilon_{\Set{C}_i}(1)=1$, the mixed strategy probability for $t=2$ becomes deterministic as
$\pi_{\Set{C}_i,j}(2)=1$ if $j=j_0$ and $\pi_{\Set{C}_i,j}(2)=0$ if $j\neq j_0$.
%
Thus, the choice of the player $\Set{C}_i$'s action remains unchanged, hereinafter.
Although this exhibits a fast convergence, the learning algorithm does not get the opportunity to explore the remaining actions and make an accurate estimation on the ensemble averages thus yielding an inaccurate steady state distribution.

In contrast, the choice of a small $\kappa$ will yield $\exp\big(\kappa\hat{r}_{\Set{C}_i,j}(1)\big)\approx\exp\big(\kappa\hat{r}_{\Set{C}_i,j'}(1)\big)$ and  $G_{\Set{C}_i,j}\big(\vect{\hat{r}}_{\Set{C}_i}(1)\big)\allowbreak\approx \frac{1}{|\Set{A}_{\Set{C}_i}|}$ for any $j,j'\in\Set{A}_{\Set{C}_i}$.
Thus, the mixed strategy probability remains as a (almost) uniform distribution based on the updating procedure given in (\ref{eqn:algoUpdates}).
Here, the learning algorithm explores all the actions indefinitely without exploiting actions with higher regrets which eventually converges to a stationary distribution with a low performance.
Therefore, from a practical implementation point of view, the choice of $\kappa$ needs to be sufficiently small for the algorithm to explore the action space and sufficiently large to improve the optimality of the solution once it converges to a stationary distribution.

\section{Simulation Results}\label{sec:results}

\begin{table}[!t]
	\centering
	\caption{Simulation parameters.}
	\label{tab:sim_para}
	\begin{tabular}{l c}
		\hline
		{\bf Parameter} & {\bf Value} \\
		\hline \hline
		Carrier frequency, System bandwidth	& $2$ GHz, $10$ MHz \\
		Thermal noise ($N_0$) &$-174$ dBm/Hz \\
		Mean traffic influx rate $\big(\eta(\vectx)\big)$ & 180 kbps\\
		Maximum transmission powers: MBS, SBS & $46,~30$ dBm \\
		Efficiency of power units ($\vartheta$): MBS, SBS~\cite{pap:georgios12B} & 23.55\%, 5.42\% \\
		Base power consumption ($P_b^{\texth{Base}}$): MBS, SBS & 40, 33 dBm \\
		Fraction of energy saved by switching OFF ($q_b$) & 0.5\\
		\hline \multicolumn{2}{c}{{\bf Minimum  distances and path loss models} ($d$ in km)~\cite{onln:3gpp10}} \\ \hline
		MBS -- SBS, MBS -- UE & 75 m, 35m \\
		SBS -- SBS, SBS -- UE & 40 m, 10 m \\
		MBS -- UE path loss& $128.1 + 37.6\log_{10}(d)$\\
		SBS -- UE path loss& $140.7 + 37.6\log_{10}(d)$\\
		\hline \multicolumn{2}{c}{{\bf Clustering}} \\ \hline
		Range of neighborhood ($\varepsilon_d$) & 250 m\\
		Impact of neighborhood width ($\sigma_d, \sigma_l$) & 300, 1 \\
		Tradeoff between similarities ($\theta$) & 0.5 \\
		Sensitivity of intra-cluster coordination cost ($\chi$) & $4.78$ dBm/m \\
		\hline \multicolumn{2}{c}{\bf Learning} \\ \hline
		Impact of load in UE association ($n$) & 1 \\
		Boltzmann temperature ($\kappa$) & 10\\
		Energy and load impacts on cost ($\lambda,~\mu$) & 0.5,~0.5 \\
		learning rate exponents for $\tau,~\iota~,\varepsilon~\mbox{and}~\nu$ & $0.6,~0.7,~0.8,~0.9$ \\
		\hline
	\end{tabular}
	\vspace{4pt}
\end{table}

For our simulations, we consider a single macrocell underlaid with an arbitrary number of SBSs and UEs uniformly distributed over the area.
All the BSs share the entire spectrum and thus, suffer from co-channel interference.
We conduct multiple simulations for various practical configurations and the presented results are averaged a large number of independent runs.
The parameters used for the simulations are summarized in Table~\ref{tab:sim_para}.
The proposed cluster-based coordination and learning based ON/OFF mechanisms are compared with the conventional network operation referred to hereinafter as ``classical approach'' in which BSs always transmit.
For further comparisons, we consider a random BS ON/OFF switching with equal probability and finally an uncoordinated learning based ON/OFF mechanism without forming clusters.
These are referred to as ``random ON/OFF" and ``learning without clusters".
As we consider three different clustering techniques, $k$-mean clustering, spectral clustering and P2P-BS clustering, they are referred to as ``learning with $k$-mean clustering", ``learning with spectral clustering", and ``learning with P2P-BS clustering" hereinafter, respectively.
For all these three clustering methods, the clusters remain unchanged for an interval of $N=100$.

\subsection{Performance of Proposed Mechanisms Based on Network Cost}\label{subsec:results_performance}

\begin{figure}[!t]
	\centering
	\includegraphics[trim = 0mm 0mm 0mm 5mm, clip,keepaspectratio,width=\myfigfactor\textwidth]{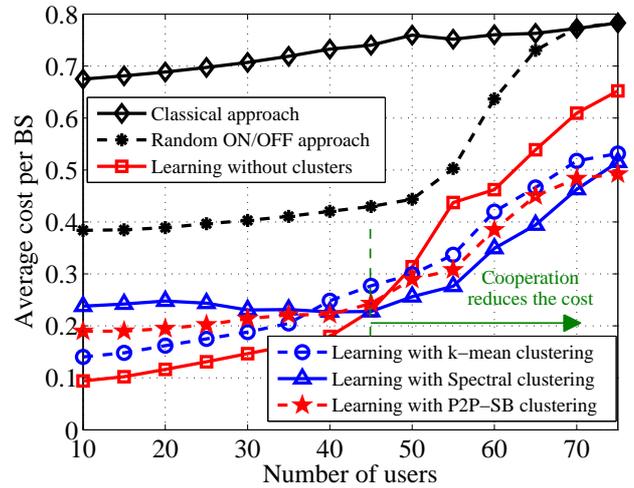}
	\vspace{-8pt}
	\caption{Average cost per BS as a function of the number of UEs with 10 SBSs. The joint similarity is used for the clustering.}
	\label{fig:cost_vs_ue}
	\vspace{-5pt}
\end{figure}

Fig. \ref{fig:cost_vs_ue} shows the changes in the average cost per BS as the number of UEs varies.
As the number of UEs increases, the BS energy consumption and load increase and thus, the average cost increases.
However, Fig.~\ref{fig:cost_vs_ue} shows that the proposed learning approaches reduce the average cost by balancing between the energy consumption and the load.
The reduction of average cost in learning without clusters, learning with $k$-mean clustering, learning with spectral clustering, and learning with P2P-BS clustering compared to the classical scenario are $60.2\%$, $60.5\%$, $59.8\%$, and $60.6\%$, respectively.
Although the random ON/OFF approach manages to reduce the cost by $31.3\%$ compared to the classical scenario, the cost reduction is not as high as compared to the learning approaches.
The random ON/OFF behavior leads to overloading the switched-ON BSs and thus, the energy saving of switched-OFF BSs is insignificant due to the increased load due to the traffic and the energy consumption of overloaded BSs.
As the number of UEs increases, BSs have lesser opportunities to switch-OFF while satisfying the UEs' QoS and thus, the behavior becomes closer to the classical approach.

In Fig.~\ref{fig:cost_vs_ue}, we can see that for small network sizes, the cluster-based approaches consume more energy than the learning approach without clustering.
This is due to the fact that, for such networks, only a small number of BSs needs to be switched-ON and, thus, the cluster-based approaches would require extra energy of $\delta P_b^{\texth{Base}}$ for coordination.
In contrast, for highly-loaded networks, as seen in Fig.~\ref{fig:cost_vs_ue}, clustering allows to better offload traffic and, subsequently, improve the overall energy efficiency.
Moreover, for highly-loaded networks, the spectral clustering approach reduces the cost compared to all three clustering mechanisms.
Although P2P-SB clustering is a decentralized clustering method, the cost reduction of learning with P2P-SB is higher than the centralized $k$-mean clustering for highly loaded networks, in which P2P-SB clustering exploits connectivity and compactness of nodes in the graph $G=(\Set{B},\Set{E})$ similar to the centralized spectral clustering method.
In this respect, Fig.~\ref{fig:cost_vs_ue} shows that spectral clustering yields, respectively, up to $48\%,~46\%,~26\%,~15\%$, and $12\%$ of cost reduction, relative to the classical approach, random ON/OFF, learning without clusters, learning with $k$-mean clustering, and learning with P2P-SB clustering, for a network with 65 UEs.

\subsection{Reductions of Energy Consumption and Time Load}

In Fig. \ref{fig:cdf_energy}, we show the CDF of the BSs' energy consumption for 10 SBSs and 50 UEs.
Fig. \ref{fig:cdf_energy} illustrates that the random ON/OFF approach achieves the highest energy consumption.
This is due to the fact that the switching OFF of some random BSs overloads the switched-ON BSs which causes a higher energy consumption.
The increased energy consumption is not compensated by the energy saving of switched-OFF BSs and thus, a higher average energy consumption per BS can be observed.
Moreover, we can see that for both classical and random ON/OFF approaches, the fraction of BSs having a high energy consumption is much higher than in cases with learning.
Indeed, the proposed learning method allows lightly-loaded BSs to offload their traffic and switch OFF, thus yielding significant energy reductions.
Coordination between clusters allows more BSs to switch OFF and, thus, as shown in Fig.~\ref{fig:cdf_energy}, the proposed learning approaches with clustering yield larger number of BSs consuming less energy.
The average energy consumption reductions of learning without clusters approach is $11\%$ compared to the classical approach.
With cluster-based coordination, for $k$-mean clustering, spectral clustering, and P2P-SB clustering, the average energy consumptions is further reduced by $18\%$, $30\%$, and $25\%$, respectively, compared to the learning without clusters.

\begin{figure}[!t]
	\centering
	\subfloat[Cumulative density function of BS energy consumption.]{
		\includegraphics[trim = 0mm 0mm 0mm 5mm, clip,keepaspectratio,width=\mysubfigfactor\textwidth]{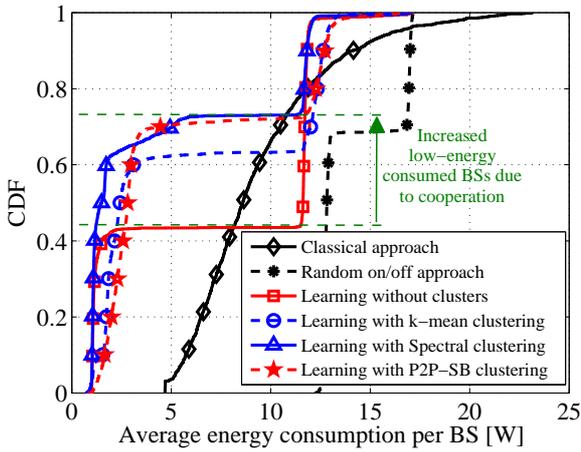}
		\vspace{-10pt}
		\label{fig:cdf_energy}
	}\hfill
	\subfloat[Cumulative density function of BS time load.]{
		\includegraphics[trim = 0mm 0mm 0mm 3mm, clip,keepaspectratio,width=\mysubfigfactor\textwidth]{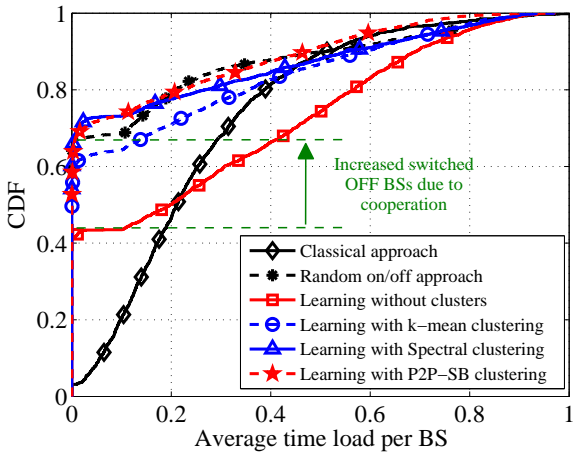}
		\vspace{-10pt}
		\label{fig:cdf_load}
	}
	\caption{Cumulative density functions of BS energy consumption and time load for 10 SBSs and 50 UEs. The joint similarity is used for the clustering.}
	\label{fig:cdfs}
	\vspace{-17pt}
\end{figure}

In Fig. \ref{fig:cdf_load}, we show the CDF of the BSs' load for 10 SBSs and 50 UEs.
Here, the random ON/OFF method exhibits a large number of BSs with low load.
Similar to the CDF of BSs' energy consumption, Fig.~\ref{fig:cdf_load} shows that the proposed learning and the dynamic clustering yield a higher number of switched-OFF BSs.
However, we can see that the average load of learning without clusters is increased by $5\%$ compared to the classical approach.
This increased load is due to the selfish behavior of BSs in learning without clusters approach.
This is due to the fact that these BSs offload their traffic without coordination with the neighboring BSs, and thus, inefficient offloading causes an increased load on average.
With the cluster-based coordination, the traffic offloading among BSs become more efficient.
Thus, $k$-mean clustering, spectral clustering and P2P-BS clustering methods yield, respectively, $34\%$, $48\%$, and $55\%$ of average load reductions compared to the classical approach.
Furthermore, Fig.~\ref{fig:cdf_load} illustrates that the proposed ON/OFF mechanism allows about $45\%$ of BSs in learning without clusters approach to switch OFF.
With the cluster-based coordination, traffic offloading becomes more efficient and thus, the number of switched-OFF BSs are further increased by $12\%-20\%$ compared to the scenario without clusters.

\subsection{Impact of Similarity and Range of Neighborhood}

\begin{figure}[!t]
	\centering
	\subfloat[Impact of similarities on the average BS cost.]{
		\includegraphics[trim = 0mm 0mm 0mm 5mm, clip,keepaspectratio,width=\mysubfigfactor\textwidth]{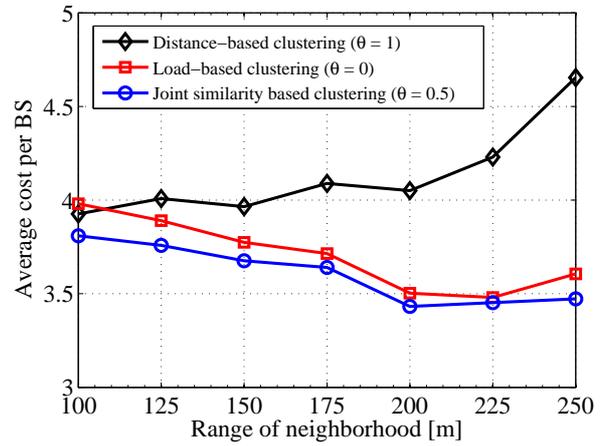}
		\vspace{-10pt}
		\label{fig:similarity_comparison}
	}\hfill
	\subfloat[Impact of similarities on the average number of clusters and the average cluster size.]{
		\includegraphics[trim = 0mm 0mm 0mm 3mm, clip,keepaspectratio,width=\mysubfigfactor\textwidth]{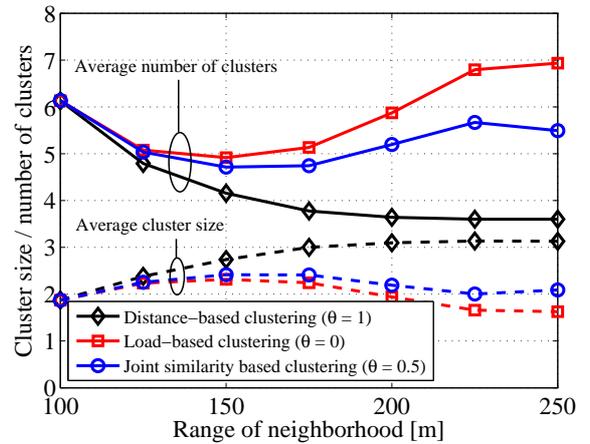}
		\vspace{-10pt}
		\label{fig:clusters_comparison}
	}
	\caption{Comparison of the effects of the similarities used for clustering with 10 SBSs and 50 UEs. The spectral clustering technique is used with the similarities calculated for $\theta = \{0, 0.5, 1\}$.}
	\vspace{-10pt}
\end{figure}

Fig. \ref{fig:similarity_comparison} shows the effect on the network cost as a function of the similarity.
The range of neighborhood ($\varepsilon_d$) defines the capability of having a physical link for the communication between BSs.
When $\theta=1$, the similarity only depends on the distance between BSs.
Thus, clustering becomes static and does not consider the BSs' loads.
As $\varepsilon_d$ increases, similarities between BSs located far from each other increases from zero to a positive quantity, and thus, larger clusters are formed with the distance-based similarity.
Since these BSs may not be able to share the traffic with one another, these large clusters become inefficient and result in increased network cost.
For $\theta=1/2$, as the knowledge of the load is combined with the distance similarity, clusters are formed with more locally-coupled BSs.
Therefore, the joint clustering results in better performance with increasing neighborhood width.
However, communication within a larger neighborhood requires more resources (cable costs and higher transmission losses for wired backhaul networks while higher transmission power and larger spectrum for over-the-air backhaul networks).
Therefore, further expanding the neighborhood increases the cost of the BSs which can be seen for $\varepsilon_d\leq 200$ m.
Clustering with only load-based similarity, $\theta=0$, displays similar behavior as the joint similarity.
The main reason is that the load-based similarity uses (\ref{eqn:neighborhood}) to determine the existence of edges.
Therefore, the limitation based on the distance implicitly appears in the load-based similarity as well.
Thus, the use of $\theta=0$ produces clusters with BSs which have sufficient capacity to support each others resulting in reduced cost compared to the clustering based on distance similarity only.

Fig.~\ref{fig:clusters_comparison} shows the average number of clusters and the average cluster sizes for the graphs with distance-based, load-based and joint similarities.
As the range of neighborhood ($\varepsilon_d$) increases, the non-zero weighted edges in $G=(\Set B, \Set E)$ increases.
Therefore, clustering based on distance similarity allows to group more BSs together thus, yielding a smaller number of clusters with large cluster size.
As the cluster size increases, the traffic in a cluster increases resulting increased loads in cluster members which directly impact the network cost.
As shown in Fig.~\ref{fig:clusters_comparison}, for $\theta<1$, up to $\varepsilon_d<175$ m, the network cost is reduced by forming large clusters.
However, for $\varepsilon_d>175$ m, increased traffic in a cluster affects the load-based similarity and thus, forming large clusters is avoided.

\subsection{Impact of Sensitivity $\chi$ on Clustering Cost }\label{subsec:overhead_cost_results}

\begin{figure}[t]
	\centering
	\includegraphics[trim = 0mm 0mm 0mm 5mm, clip,keepaspectratio,width=\myfigfactor\textwidth]{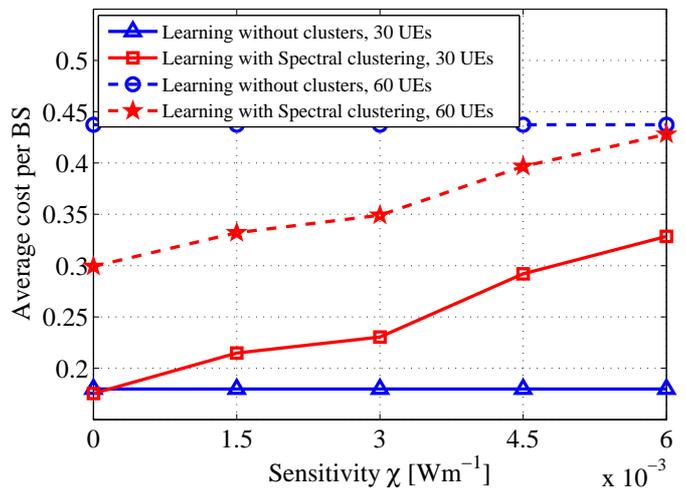}
	\vspace{-10pt}
	\caption{Comparison of the average cost per BS for learning with and without clusters for different choices of the sensitivity ($\chi$) parameter. The setup is for 10 SBSs and  $\varepsilon_d=250~\text{m}$.}
	\label{fig:cost_sensitivity}
\end{figure}

Fig.~\ref{fig:cost_sensitivity} shows the average cost per BS for learning without clusters and learning with spectral clustering for different choices of the sensitivity parameter $\chi$ in (\ref{eqn:overheadcost_model}).
The sensitivity parameter $\chi$ models how sensitive the clusters' power consumption is with respect to the overhead of their intra-cluster coordination.
The choices of these parameters are such that
\emph{i)} $\chi=0$ will lead to an overhead free (unrealistic) scenario,
\emph{ii)} $\chi=3\times10^{-3}$ implies $\delta P_b^{\texth{Base}}\approx \frac{1}{2}P_b^{\texth{Base}}$,
\emph{iii)} $\chi=6\times10^{-3}$ implies $\delta P_b^{\texth{Base}}\approx P_b^{\texth{Base}}$,
for $\varepsilon_d=250~\text{m}$ and the average cluster size at the above neighborhood range.
Since learning without clusters has no effect on this sensitivity due to the absence of clusters and intra-cluster coordination, the average cost remains unchanged for different choices of $\chi$.
For learning with spectral clustering, as $\chi$ increases, clusters need higher power consumption for their intra-cluster coordination and thus, the average cost per BS is increased.
According to (\ref{eqn:overheadcost_model}), intra-cluster coordination cost is independent of the number of UEs, and Fig.~\ref{fig:cost_sensitivity} shows that the changes in $\chi$ have the same effects on both systems with 30 UEs and 60 UEs.
It can be noted that the benefit of clustering to reduce the average cost increases with the number of UEs compared to learning without clustering method, as explained under Fig.~\ref{fig:cost_vs_ue}.

\section{Conclusions}\label{sec:conclusions}

In this paper, we have proposed a dynamic cluster-based ON/OFF mechanism for small cell base stations.
Clustering allows clustered base stations to coordinate their transmission while the clusters compete with one another to reduce a per-cluster cost based on their energy consumption and time load due to their traffic.
In this regard we have formulated the per-cluster cost minimization problem as a noncooperative game among clusters.
To solve the game, we have proposed a distributed algorithm and an intra-cluster coordination method using which base stations choose their transmission modes with minimum overhead.
Our proposed clustering method uses information on both the locations of BSs and their capability of handling the traffic and dynamically form the clusters in order to improve the overall performance.
For clustering, both centralized and decentralized clustering techniques are introduced and the performances are compared.
Simulation results have shown significant gains in energy expenditure and load reduction compared to the conventional transmission techniques.
Moreover, the results provide an insight of when and how to reap the benefits from the cluster-based coordination in small cell networks.

\bibliographystyle{IEEEtran}
\bibliography{myBibliography_jsw}

\begin{IEEEbiography}[{\includegraphics[width=1in,height=1.25in,clip,keepaspectratio]{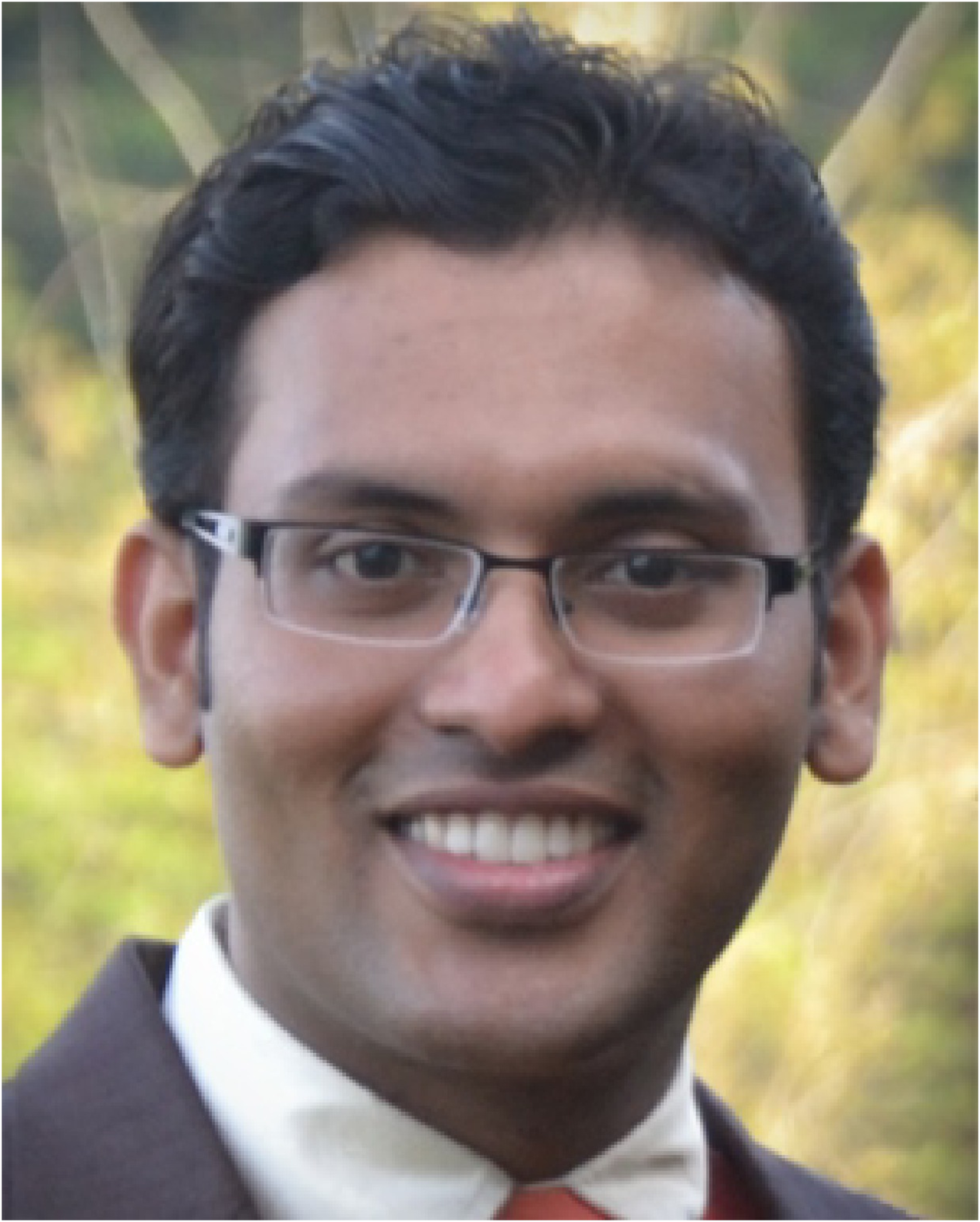}}]{Sumudu Samarakoon}
	received his B. Sc. (Hons.) degree in Electronic and Telecommunication Engineering from the University of Moratuwa, Sri Lanka and the M. Eng. degree from the Asian Institute of Technology, Thailand in 2009 and 2011, respectively.
	He is currently persuading Dr. Tech degree in Communications Engineering in University of Oulu, Finland.
	Sumudu is also a member of the research staff of the Centre for Wireless Communications (CWC), Oulu, Finland.
	His main research interests are in heterogeneous networks, radio resource management, machine learning, and game theory.
\end{IEEEbiography}

\begin{IEEEbiography}[{\includegraphics[width=1in,height=1.25in,clip,keepaspectratio]{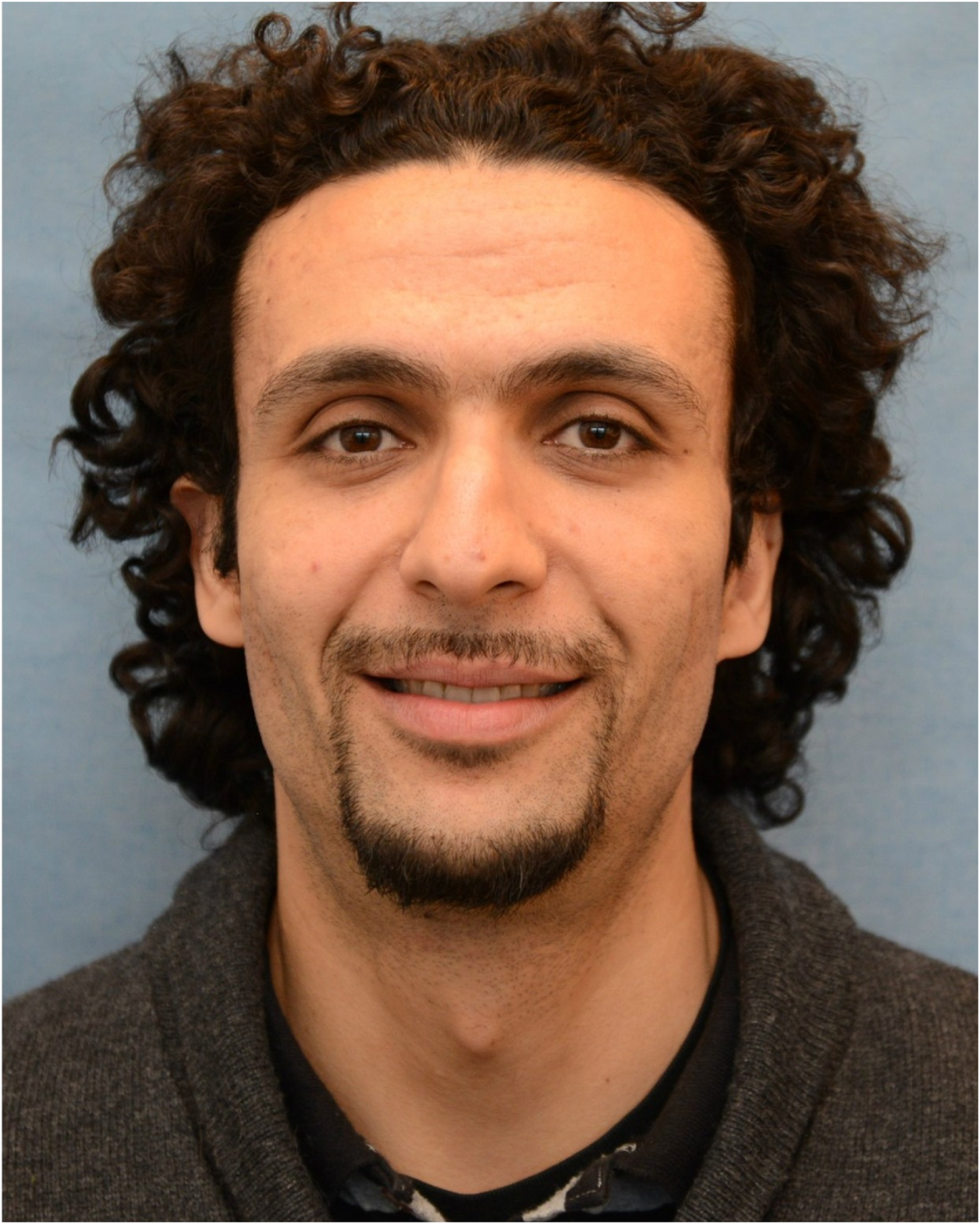}}]{Mehdi Bennis}
	received his M.Sc. degree in Electrical Engineering jointly from the EPFL, Switzerland and the Eurecom Institute, France in 2002.
	From 2002 to 2004, he worked as a research engineer at IMRA-EUROPE investigating adaptive equalization algorithms for mobile digital TV.
	In 2004, he joined the Centre for Wireless Communications (CWC) at the University of Oulu, Finland as a research scientist.
	In 2008, he was a visiting researcher at the Alcatel-Lucent chair on flexible radio, SUPELEC.
	He obtained his Ph.D. in December 2009 on spectrum sharing for future mobile cellular systems.
	
	His main research interests are in radio resource management, heterogeneous networks, game theory and machine learning in 5G networks and beyond.
	He has co-authored one book and published more than 100 research papers in international conferences, journals and book chapters.
	He was the recipient of the prestigious 2015 Fred W. Ellersick Prize from the IEEE Communications Society.
	Dr. Bennis serves as an editor for the IEEE Transactions on Wireless Communications.
\end{IEEEbiography}

\begin{IEEEbiography}[{\includegraphics[width=1in,height=1.25in,clip,keepaspectratio]{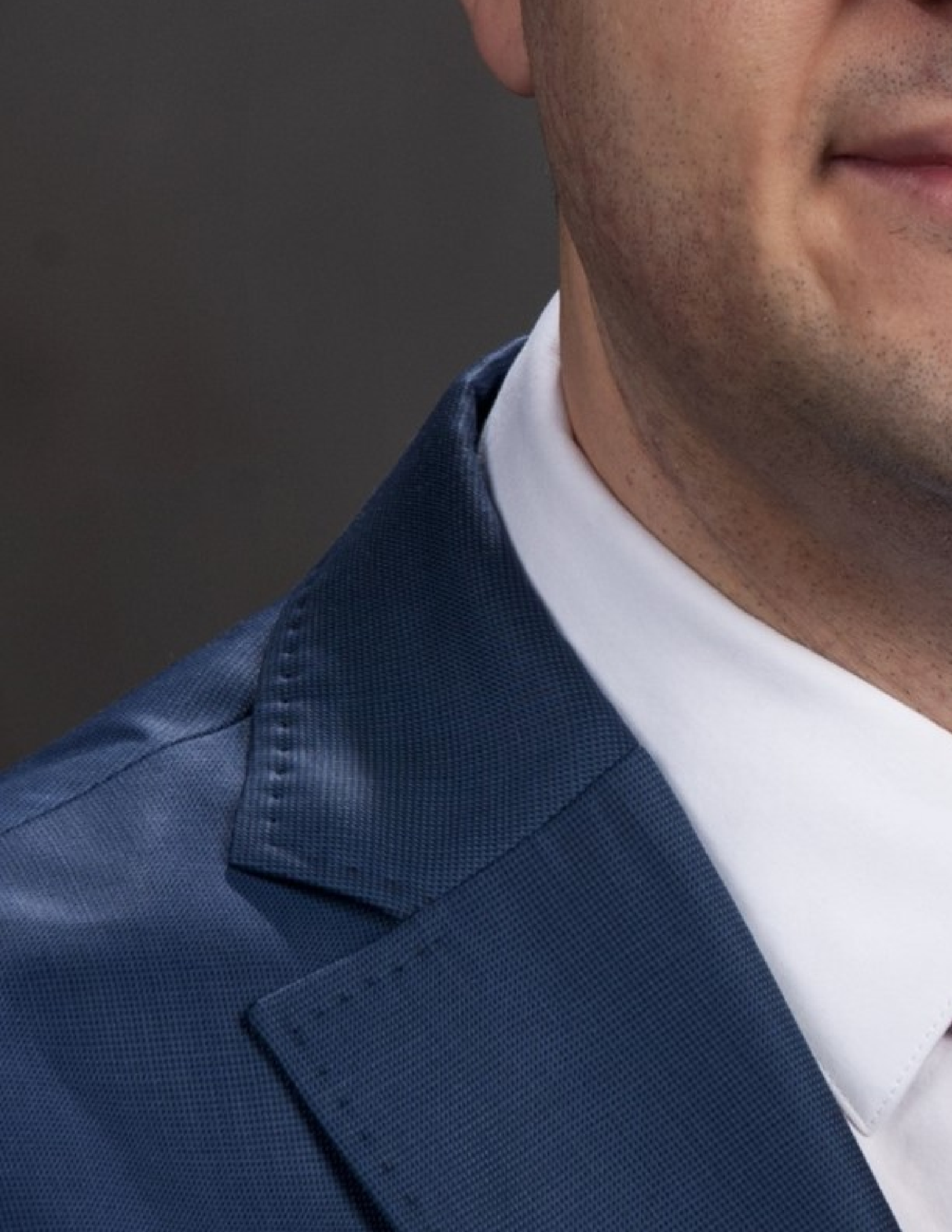}}]{Walid Saad}
	(S'07, M'10, SM’15) Walid Saad received his Ph.D degree from the University of Oslo in 2010.
	Currently, he is an Assistant Professor at the Bradley Department of Electrical and Computer Engineering at Virginia Tech, where he leads the Network Science, Wireless, and Security (NetSciWiS) laboratory, within the Wireless@VT research group.
	His  research interests include wireless and social networks, game theory, cybersecurity, and cyber-physical systems.
	Dr. Saad is the recipient of the NSF CAREER award in 2013, the AFOSR summer faculty fellowship in 2014, and the Young Investigator Award from the Office of Naval Research (ONR) in 2015.
	He was the author/co-author of four conference best paper awards at WiOpt in 2009, ICIMP in 2010, IEEE WCNC in 2012, and IEEE PIMRC in 2015.
	He is the recipient of the 2015 Fred W. Ellersick Prize from the IEEE Communications Society.
	Dr. Saad serves as an editor for the IEEE Transactions on Wireless Communications, IEEE Transactions on Communications, and IEEE Transactions on Information Forensics and Security.
\end{IEEEbiography}

\begin{IEEEbiography}[{\includegraphics[width=1in,height=1.25in,clip,keepaspectratio]{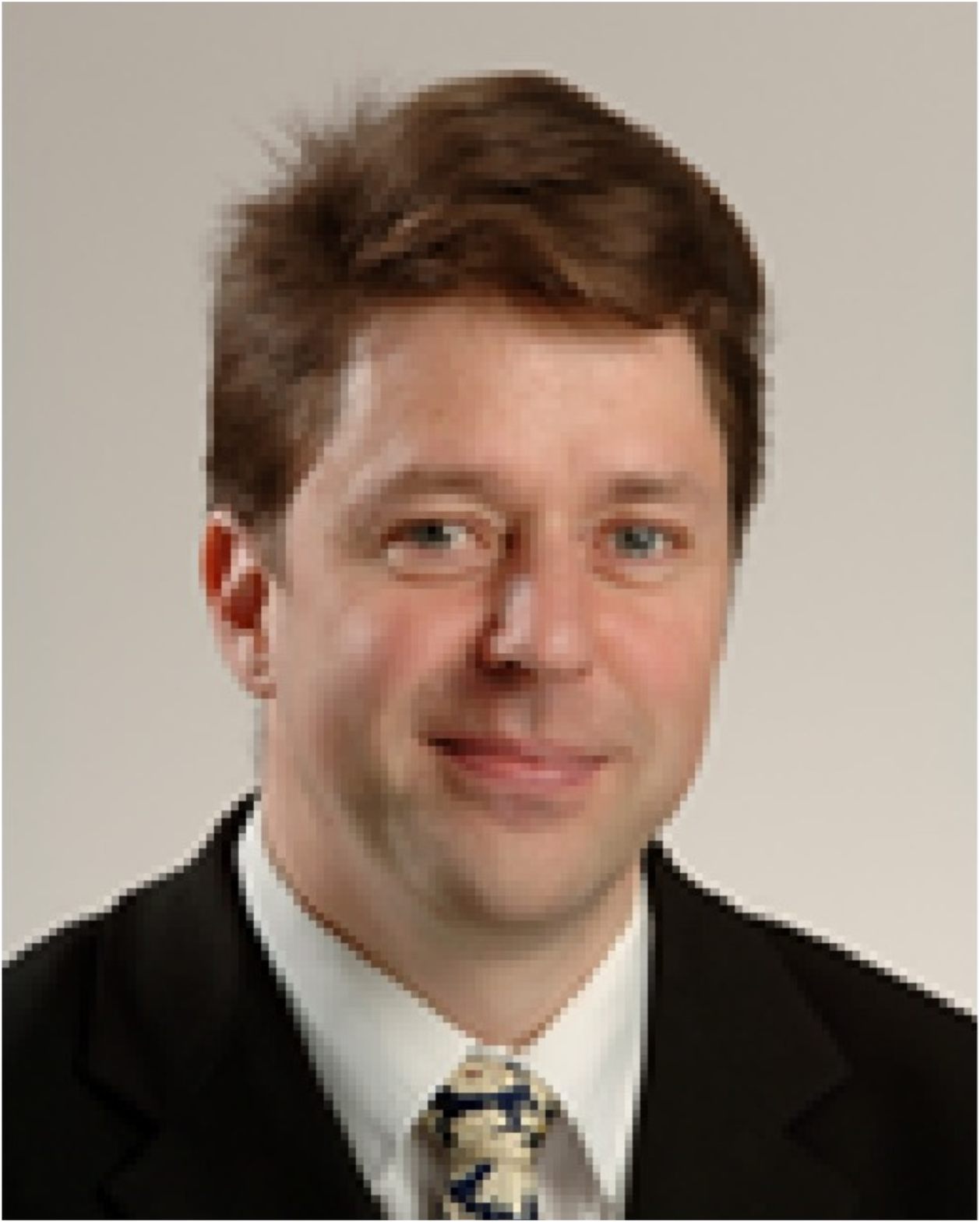}}]{Matti Latva-aho}
	received the M.Sc., Lic.Tech. and Dr. Tech (Hons.) degrees in Electrical Engineering from the University of Oulu, Finland in 1992, 1996 and 1998, respectively.
	From 1992 to 1993, he was a Research Engineer at Nokia Mobile Phones, Oulu, Finland after which he joined Centre for Wireless Communications (CWC) at the University of Oulu.
	Prof. Latva-aho was Director of CWC during the years 1998-2006 and Head of Department for Communication Engineering until August 2014.
	Currently he is Professor of Digital Transmission Techniques at the University of Oulu.
	His research interests are related to mobile broadband communication systems and currently his group focuses on 5G systems research.
	Prof. Latva-aho has published 300+ conference or journal papers in the field of wireless communications.
\end{IEEEbiography}







\end{document}